\RequirePackage{etex}

\documentclass[12pt,psamsfonts]{amsart}
\usepackage{amsfonts,pstricks}
\usepackage{amsmath, amsthm, amssymb,  mathrsfs, epsfig}
\usepackage{dcpic, pictex, pinlabel}
\usepackage{hyperref}
\usepackage{graphicx}
\usepackage{color}
\usepackage{ifpdf}
\usepackage{tikz}
\usepackage{xypic}
\usepackage{extarrows}
\usepackage{multirow}

\DeclareMathOperator{\diag}{diag}

\begin{document}

\newcommand{\Z}{\mathbb{Z}}
\newcommand{\M}{\mathcal{M}}
\newcommand{\B}{\mathcal{B}}
\newcommand{\Sym}{\mathcal{S}}
\newcommand{\A}{\mathcal{A}}
\theoremstyle{definition}
\newtheorem{axiom}{Axiom}
\newtheorem{thm}{Theorem}
\newtheorem{Conjecture}{Conjecture}
\newtheorem{lem}{Lemma}
\newtheorem{example}{Example}
\newtheorem{cor}{Corollory}
\newtheorem{prop}{Proposition}
\newtheorem{rem}{Remark}
\newtheorem{definition}{Definition}
\newtheorem{measurement}{Measurement}
\newtheorem{ancilla}{Ancilla}

\numberwithin{equation}{section} \makeatletter
\renewenvironment{proof}[1][\proofname]{\par
    \pushQED{\qed}%
    \normalfont \topsep6\p@\@plus6\p@ \labelsep1em\relax
    \trivlist
    \item[\hskip\labelsep\indent
        \bfseries #1]\ignorespaces
}{%
    \popQED\endtrivlist\@endpefalse
} \makeatother
\renewcommand{\proofname}{Proof}

\title{Universal quantum computation with weakly integral anyons}
\author{Shawn X. Cui$^{1}$, SEUNG-MOON HONG$^2$, and Zhenghan Wang$^{1,3}$}

\address{$^1$Department of Mathematics\\University of California\\Santa Barbara, CA 93106}
\email{xingshan@math.ucsb.edu, zhenghwa@math.ucsb.edu}
\address{$^2$Department of Mathematics and Statistics, University of Toledo,
Toledo, OH 43606}
\email{SeungMoon.Hong@utoledo.edu}
\address{$^3$Microsoft Research, Station Q\\ University of California\\ Santa Barbara, CA 93106}
\email{zhenghwa@microsoft.com}

\thanks{The third author thanks A. Kitaev for explaining his encoding.  S.-X. C and Z.W. are partially supported by NSF DMS 1108736.}

\keywords{anyonic quantum computation, universal gate set, braid group}

\date{}

\begin{abstract}
Harnessing non-abelian statistics of anyons to perform quantum computational tasks is getting closer to reality. While the existence of universal anyons by braiding alone such as the Fibonacci anyon is theoretically a possibility, accessible anyons with current technology all belong to a class that is called weakly integral---anyons whose squared quantum dimensions are integers. We analyze the computational power of the first non-abelian anyon system with only integral quantum dimensions---$D(S_3)$, the quantum double of $S_3$. Since all anyons in $D(S_3)$ have finite images of braid group representations, they cannot be universal for quantum computation by braiding alone. Based on our knowledge of the images of the braid group representations, we set up three qutrit computational models. Supplementing braidings with some measurements and ancillary states, we find a universal gate set for each model.
\end{abstract}

\maketitle

\section{Introduction}

Harnessing non-abelian statistics of anyons to perform quantum computational tasks is getting closer to reality.  While the existence of universal anyons by braiding alone such as the Fibonacci anyon is theoretically a possibility \cite{mong13,RR99}, accessible anyons with current technology all belong to a class that is called weakly integral (WI)---anyons whose squared quantum dimensions are integers.  A famous WI anyon is the Ising anyon $\sigma$ with $d_\sigma=\sqrt{2}$, which is believed to model the non-abelian quasi-particle in the fractional quantum Hall liquids at $\nu=\frac{5}{2}$ \cite{nayak}.  Other WI anyons with Property F \cite{RDepak} include the metaplectic anyons \cite{metaplectic13p,metaplectic13m} and those in quantum double of finite groups \cite{ERW}.  Certain topological defects or ends of $1D$ nanowires also behave as WI anyons \cite{genon13,kitaev2001unpaired}.  It is conjectured that all WI anyons have finite images of braid group representations \cite{RDepak}, if so then they cannot be universal for quantum computation by braiding alone.

In this paper, we analyze the computational power of the anyon system  $D(\Sym_3)$---the quantum double of $\Sym_3$.  $D(\Sym_3)$ is the first non-abelian anyon system with only integral quantum dimensions \cite{OnMC}.  There are $8$ anyon types in the theory, which are denoted by  $A,B,C,D,E,F,G,H$ as in \cite{BSW} with quantum dimensions $\{1,1,2,3,3,2,2,2\}$.  It is known that all braid images are finite \cite{ERW}.  It follows that to obtain a universal quantum gate set, we have to go beyond braiding.  The natural extra resources are measurements and ancillary states.  Using measurement to gain extra computational power is tricky because universal quantum computation can be performed by measurement alone.  Similar caution applies to ancillary states as cluster state quantum computation shows.  Therefore, we have to be careful in choosing physically reasonable extra resources from measurements and ancillary states.

In \cite{Kitaev2} (see also \cite{Mochon}), the anyon of quantum dimension $d_D=3$, denoted as $D$, is made universal for quantum computation by encoding a qutrit in the $A,C$ fusion channels of a pair of $D$ anyons.  In the usual anyonic quantum computing model, two fusion channels such as $A,C$ of a pair of $D$ anyons would be used to encode a qubit instead of a a qutrit, but the authors split the single fusion channel $C$ into a two dimensional internal computational space because the anyon $C$ is the $2$-dimensional representation of $\Sym_3$.  Since $D(\Sym_3)$ is a discrete gauge theory, this encoding is justified on physical ground and computation is performed using only the color (topological) degrees of freedom of anyons.
Universality based on similar encodings for other finite groups is established in \cite{preskill, mochon2}.  In this paper, we follow the usual scheme in anyonic quantum computation by encoding information in the subspace of anyonic fusion tree basis without splitting any fusion channels, i.e.,  the computational subspace is spanned by basis elements from labeling a single fusion tree.  Since a pair of $D$ anyons has $9$ fusion channels, we have many choices of encoding a qutrit by choosing a $3$-dimensional fusion subspace.  Based on our analysis of the representations of the braid groups, we propose three different encodings of qutrits: one with $3$ fusion tree bases, and two with superpositions of fusion tree bases.  It is possible, as suggested by Kitaev to the third author, that the splitting of the fusion channel $C$ as in \cite{Kitaev2} can be understood as a non-local encoding using superpositions of different fusion trees.  Note that in the encoding in \cite{Kitaev2}, there is a bureau of standards, which is mathematically a based frame.

The contents of the paper is as follows. In Section $2$, we provide the detail of our adaptive anyonic quantum computing models and prove that a qutrit gate set convenient for our purpose is universal.  There are three natural choices in our set-up that are called $U$-model, $V$-model and $W$-model, respectively.  We also define the measurements and ancillary states that we are going to use later.  In Section $3$,  we prove that the $U$- and $V$- models are universal when braidings are supplemented by the two measurements defined in Section $2$, and the $W$-model needs the extra ancillary state to become universal.  Our major technical advance is organized into two appendices.  In Appendix $A$,
 we obtain the complete solutions of modular categories with the same fusion rules as $D(\Sym_3)$. To save space, we list only the complete data for the $D(\Sym_3)$.  For this paper, any other theory will work equally well.  The complete list of data in Appendix $A$ is used in Appendix $B$ to analyze the images of the braid group representations of $\B_4$.  We give complete information of the finite images as abstract groups, and as concretes matrices with respect to the computational bases of our models.  These matrices are braiding quantum circuits of our models. Many interesting finite groups such as the Hessian group of order $216$ appeared as images of braid group representations.

\section{Adaptive anyonic quantum computing model}\label{background}
A pure anyonic quantum computing model as illustrated by Fig. $7.1$ in \cite{Wang} is to implement a circuit by braiding alone.  Measurement is only done in the end by fusing anyons together.  In particular, we are not allowed to do measurements in the middle of the computation.  Unlike the standard circuit model, a computation with measurements during the computing process is not always equivalent to one that all measurements are postponed to the end of the computation.  Since WI anyons provide only very limited circuits by braiding alone, we have to rely on other resources to obtain a universal gate set.  The obvious places to look for are measurements in the middle of the computation and ancillary states.  Since measurements of anyon charges beyond fusing two anyons are subtle, we want to do as little measurement as possible so that we do not decohere or leak the protected information in the computational subspace.  In this section, we use $D(\Sym_3)$ to illustrate such adaptive models that braiding gates are supplemented with measurements and ancillary states.  Our goal is to find minimal extra resources beyond braidings to obtain a universal gate set.  Therefore, though important for physically realizing the extra resources, we will not justify our choices of measurements or ancillary states physically.  Another important issue that we did not address is the issue of leakage.  We think there is no damaging leakage in our model, but will leave a careful analysis to the future.

\subsection{The integral anyon system $D(\Sym_3)$}
The irreducible representations (irreps) of the quantum group $D(\Sym_3)$, called the Drinfeld double or quantum double of $\Sym_3$, correspond to pairs $(C, \rho_u),$ where $C$ is a conjugacy class of $\Sym_3,$  $u$ an element in $C$, and $\rho_u$ an irrep of the centralizer of $u$ in $\Sym_3.$ For a fixed conjugacy class $C$, the corresponding irreps of $D(\Sym_3)$ do not depend on the choice of the element $u$. There are three conjugacy classes of $\Sym_3$, namely $C_1 = \{e\},\; C_2 = \{(12),(23),(13)\},\; C_3 = \{(123),(132)\}.$ For $C_1$, the centralizer of $e$ is $\Sym_3$, which has three irreps, i.e. the trivial one, the sign one, and the $2$-dimensional one.  We denote them by $A$, $B$ and $C$, respectively. For $C_2,$ we pick $(12)$ and its centralizer is isomorphic to $\mathbb{Z}_2,$ which has two irreps. We denote the trivial one by $D$, and the other one by $E$. For $C_3,$ the centralizer of $(123)$ is isomorphic to $\mathbb{Z}_3,$ which has three irreps, all of which are $1$-dimensional. More precisely, they correspond to mapping the generator of $\mathbb{Z}_3$ to $1$, $\omega=e^{\frac{2\pi i}{3}}=-\frac{1}{2}+\frac{\sqrt{3}i}{2} $ and $\omega^2.$ We denote these three irreps by $F$, $G$ and $H$, respectively. Therefore, there are in total $8$ irreps of $D(\Sym_3).$ As a unitary modular category, $Rep(D(\Sym_3))$ has $8$ isomorphism classes of simple objects. Since simple objects in unitary modular categories models anyons, we also call them anyon types. The $8$ anyon types were denoted by $A$, $B$, $C$, $D$, $E$, $F$, $G$ and $H$.  In the discrete gauge theory, an anyon of type $A$ is called the vacuum; anyons of types $B$ and $C$ are purely electric charges; anyons of types $D$ and $F$ are purely magnetic fluxes; while anyons of types $E$, $G$ and $H$ are dyons. We will not always distinguish between anyon types (isomorphism classes of simple objects) and anyons (simple objects) carefully because for $D(\Sym_3)$ this distinction will not make any difference.  Detailed explanations of the quantum double of $\Sym_3$ can be found in many references, e.g., \cite{Kitaev} \cite{BSW}.

We list the irreps of $D(\Sym_3)$ and their quantum dimensions in Table \ref{conjugacyclass}, and the fusion rules in Table \ref{fusionrule}.

\begin{table}\caption{$D(\Sym_3)$}\label{conjugacyclass}
\begin{tabular}{|c|c|c|c|}
\hline
Flux (Conjugacy class) & Centralizer  &  Charge & qdim\\
\hline
$C_1=\{e\}$ & $Z(e)=\Sym_3$ & 1, -1, 2&1,1,2 \\
\hline
\multirow{2}{*}{$C_2=\{(12),(13), (23)\}$} & $Z((12))=\mathbb{Z}_2$ &\multirow{2}{*}{ +, - }&\multirow{2}{*}{3,3}\\
                                           &   $=\{e, (12)\}$       &                       &                    \\
\hline
\multirow{2}{*}{$C_3=\{(123),(132)\}$ }    & $Z((123))=\mathbb{Z}_3$ & \multirow{2}{*}{1, $\omega, \bar{\omega}$}&\multirow{2}{*}{2,2,2}\\
                                           &  $=\{e, (123),(132)\}$  &                                           &        \\
\hline
\end{tabular}
\end{table}

{\tiny
\begin{table}\caption{Fusion rules of $D(\Sym_3)$}\label{fusionrule}
\begin{tabular}{|c|c|c|c|c|c|c|c|c|}
\hline $\otimes$ &$A$ &$B$ &$C$ &$D$ &$E$ &$F$ &$G$ &$H$\\ \hline
$A$ &$A$ &$B$ &$C$ &$D$ &$E$& $F$ &$G$ &$H$\\ \hline
$B$ &$B$ &$A$ &$C$& $E$ &$D$ &$F$ &$G$ &$H$\\ \hline
$C$ &$C$ &$C$ &$A\oplus B\oplus C$& $D\oplus E$ &$D\oplus E$ & $G\oplus H$& $F\oplus H$ &$F\oplus G$\\ \hline
\multirow{2}{*}{$D$} &\multirow{2}{*}{$D$} &\multirow{2}{*}{$E$} &\multirow{2}{*}{$D\oplus E$}& $A\oplus C\oplus  F$ & $B\oplus C\oplus F$ & \multirow{2}{*}{$D\oplus E$} & \multirow{2}{*}{$D\oplus E$} & \multirow{2}{*}{$D\oplus E$} \\
& & & & $\oplus G\oplus H$ & $\oplus G\oplus H$ & & &  \\ \hline
\multirow{2}{*}{$E$} &\multirow{2}{*}{$E$}& \multirow{2}{*}{$D$}& \multirow{2}{*}{$D\oplus E$} & $B\oplus C\oplus F$ & $A\oplus C\oplus F$ & \multirow{2}{*}{$D\oplus E$} &\multirow{2}{*}{$D\oplus E$} & \multirow{2}{*}{$D\oplus E$} \\
& & & & $\oplus G\oplus H$ & $\oplus G\oplus H$ & & &  \\  \hline
$F$ &$F$ & $F$& $G\oplus H$& $D\oplus E$ & $D\oplus E$ & $A\oplus B\oplus F$ & $H\oplus C$ & $G\oplus C$ \\ \hline
$G$ &$G$ & $G$& $F\oplus H$ & $D\oplus E$ & $D\oplus E$ & $H\oplus C$ & $A\oplus B\oplus G$ & $F\oplus C$ \\ \hline
$H$ &$H$ & $H$& $F\oplus G$ & $D\oplus E$ & $D\oplus E$ & $G\oplus C$ & $F\oplus C$ & $A\oplus B\oplus H$\\\hline
\end{tabular}
\end{table}
}

The modular representation of $\textrm{SL}(2,\Z)$ is given by:
 $$ T = \diag (1,1,1, -1,1,1,\omega,\omega^2).$$

$${S}=\frac{1}{6}\begin{pmatrix}
1&1 &2 &3 &3 &2 &2 &2 \\
1&1 &2 &-3 &-3 &2 &2 &2 \\
2&2 &4 &0 &0 &-2 &-2 &-2\\
3&-3 &0 &3 &-3 &0 &0 &0 \\
3&-3 &0 &-3 &3 &0 &0 &0 \\
2&2 &-2 &0 &0 &4 &-2 &-2 \\
2&2 &-2 &0&0 &-2 &-2 &4 \\
2&2 &-2 &0 &0 &-2 &4 &-2
\end{pmatrix}$$

\subsection{Fusion tree basis} \label{FusionTreeBasis}

An anyon $c$ can split into a pair of anyons $(a,b)$ if the triple $(a,b,c)$ is admissible \cite{Wang}. We denote this process by

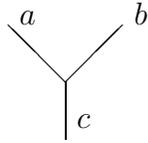
\setlength{\unitlength}{0.030in}
\begin{picture}(50,30)(-40,0)
 \put(10,10){\line(0,-1){10}}
 \put(10,10){\line(1,1){10}}
 \put(10,10){\line(-1,1){10}}

 \put(2,20){$a$}
 \put(22,20){$b$}
 \put(12,2){$c$}
\end{picture}

The anyon $a$ can continue to split into another pair of anyons. Consider the following splitting tree:

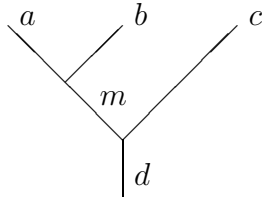
\setlength{\unitlength}{0.030in}
\begin{picture}(50,40)(-40,0)
 \put(20,10){\line(0,-1){10}}
 \put(20,10){\line(1,1){20}}
 \put(20,10){\line(-1,1){20}}
\put(10,20){\line(1,1){10}}

 \put(2,30){$a$}
 \put(22,30){$b$}
 \put(42,30){$c$}
 \put(16,16){$m$}
 \put(22,2){$d$}
\end{picture}

There is a Hilbert space $V_{d}^{abc}$ for the $4$ anyons $a,b,c,d$, where the labeled splitting trees with choices of anyon $m$ that make the splitting tree admissible at each trivalent vertex form a basis of $V_d^{a,b,c}$.  We imagine the splitting process as going from the bottom to the top. then the fusing process going from the top to the bottom.  Therefore, we will often also refer to a splitting tree as a fusion tree.

We also have another splitting tree.

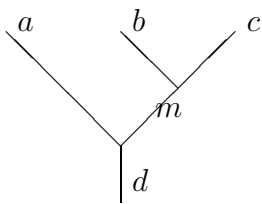
\setlength{\unitlength}{0.030in}
\begin{picture}(50,40)(-40,0)
 \put(20,10){\line(0,-1){10}}
 \put(20,10){\line(1,1){20}}
 \put(20,10){\line(-1,1){20}}
\put(30,20){\line(-1,1){10}}

 \put(2,30){$a$}
 \put(22,30){$b$}
 \put(42,30){$c$}
 \put(26,15){$m$}
 \put(22,2){$d$}
\end{picture}

The labeled fusion trees provide another basis for the same Hilbert space $V_{d}^{abc}$. Hence there is a transformation matrix $F$ that relates these two bases, which is called an $F$-matrix.

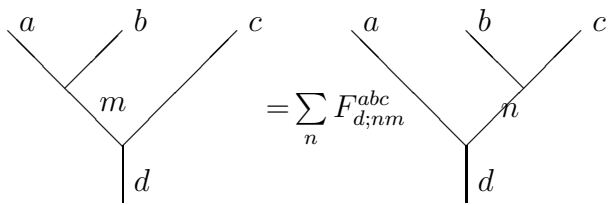
\setlength{\unitlength}{0.030in}
\begin{picture}(50,40)(0,0)
 \put(20,10){\line(0,-1){10}}
 \put(20,10){\line(1,1){20}}
 \put(20,10){\line(-1,1){20}}
\put(10,20){\line(1,1){10}}

 \put(2,30){$a$}
 \put(22,30){$b$}
 \put(42,30){$c$}
 \put(16,16){$m$}
 \put(22,2){$d$}

 \put(45,15){=}
 \put(50,15){$\sum\limits_{n} F_{d;nm}^{abc}$}

 \put(80,10){\line(0,-1){10}}
 \put(80,10){\line(1,1){20}}
 \put(80,10){\line(-1,1){20}}
\put(90,20){\line(-1,1){10}}

 \put(62,30){$a$}
 \put(82,30){$b$}
 \put(102,30){$c$}
 \put(86,15){$n$}
 \put(82,2){$d$}
\end{picture}

Assume the fusion tree on the left hand side of the above equation is the $j_m$-th basis element and the fusion tree on the right hand side is the $i_n$-th basis element in the other basis. Then $F_{d;nm}^{abc}$ will denote the $(i_n,j_m)$-entry of $F_{d}^{abc}$ when the theory has no multiplicities in the fusion rules such as $D(\Sym_3)$. The numerical values $F_{d;nm}^{abc}$ are called $6j$ symbols.  In the following, we always assume there are no multiplicities in the fusion rules, i.e., the fusion coefficients are either $0$ or $1$. The matrices $F_{d}^{abc}$ can be chosen to be unitary for a unitary theory.

\begin{rem}
For the case of $D(\Sym_3)$, recall that the 8 anyon types are denoted by $A,\; B ,\; C,\;D,\;E,\;F,\;G$ and $H$. For some reason, we order the basis elements in the order $A,\; B ,\; G,\;D,\;E,\;F,\;C, \; H$. This is the order we use to compute the $F$-matrices in Appendix \ref{FRmatrices}. For example, the basis for the following fusion tree is $\{A,\;G ,\;F ,\;C,\;H\}.$

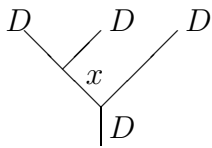
\setlength{\unitlength}{0.020in}
\begin{picture}(50,40)(-40,0)
 \put(20,10){\line(0,-1){10}}
 \put(20,10){\line(1,1){20}}
 \put(20,10){\line(-1,1){20}}
\put(10,20){\line(1,1){10}}

 \put(-5,30){$D$}
 \put(22,30){$D$}
 \put(42,30){$D$}
 \put(16,16){$x$}
 \put(22,2){$D$}
 \end{picture}
\end{rem}

The Hilbert space $V_{c}^{ab}$ associated to the following splitting tree is $1$-dimensional.

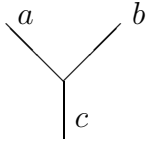
\setlength{\unitlength}{0.030in}
\begin{picture}(50,30)(-40,0)
 \put(10,10){\line(0,-1){10}}
 \put(10,10){\line(1,1){10}}
 \put(10,10){\line(-1,1){10}}

 \put(2,20){$a$}
 \put(22,20){$b$}
 \put(12,2){$c$}
\end{picture}

Braiding the two anyons $a,b$ corresponds to a unitary transformation $B_{c}^{ab}$ from $V_c^{ab}$ to $V_c^{ba}$. The image of the fusion tree basis of $V_c^{ab}$ under $B_{c}^{ab}$ is a scalar multiple of the chosen fusion tree basis of $V_c^{ba}$.  This scalar is denoted by $R_c^{ba},$ which is called the $R$-symbol.

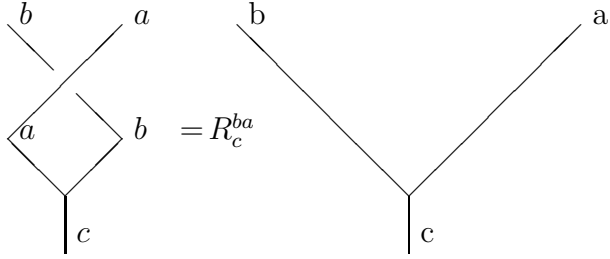
\setlength{\unitlength}{0.030in}
\begin{picture}(50,60)(0,0)
 \put(10,10){\line(0,-1){10}}
 \put(10,10){\line(1,1){10}}
 \put(10,10){\line(-1,1){10}}
 \put(0,20){\line(1,1){20}}
 \put(20,20){\line(-1,1){8}}
 \put(8,32){\line(-1,1){8}}

 \put(2,20){$a$}
 \put(22,20){$b$}
 \put(12,2){$c$}
 \put(2,40){$b$}
 \put(22,40){$a$}

 \put(30,20){=}
 \put(35,20){$R_c^{ba}$}

 \put(70,10){\line(0,-1){10}}
 \put(70,10){\line(1,1){30}}
 \put(70,10){\line(-1,1){30}}

 \put(42,40){b}
 \put(102,40){a}
 \put(72,2){c}
\end{picture}

The $6j$-symbols and $R$-symbols are the data needed to compute the amplitudes of creating anyons from the vacuum, braiding some of them, and then fusing them back to the vacuum. The approximation of these probabilities for such processes is the output of anyonic quantum computational models.  Similar orthonormal basis exists for any Hilbert space $V_{x}^{ab...c}$.  Protected information is encoded into some subspaces of $V_{x}^{ab...c}$, which are called computational subspaces.  There are no canonical choices of computational subspaces.

\subsection{Encoding of qudits}\label{UVWmodel}

Consider the following fusion tree:

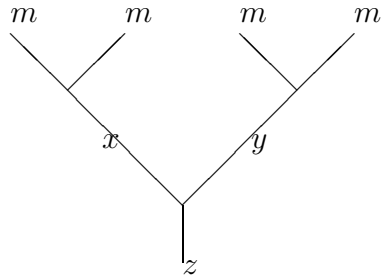
\setlength{\unitlength}{0.030in}
\begin{picture}(160,60)(0,-10)

\put(50,0){\line(0,1){10}}
\put(50,10){\line(1,1){30}}
\put(50,10){\line(-1,1){30}}
\put(70,30){\line(-1,1){10}}
\put(30,30){\line(1,1){10}}

\put(20,42){$m$}
\put(40,42){$m$}
\put(60,42){$m$}
\put(80,42){$m$}
\put(36,20){$x$}
\put(62,20){$y$}
\put(50,-2){$z$}
\end{picture}

Namely, we start with an anyon of type $z$ and split it into $4$ anyons, all of which have the same anyon type $m$.  All the pairs of $(x,y)$ that make the above splitting tree admissible form a natural basis of $V_z^{mmmm}.$ We denote them by $\{|x,y; m,z \rangle\}$. When there is no confusion, we will use the abbreviation $|xy\rangle$ for $|x,y; m,z \rangle$.  We use this basis of the Hilbert space $V_z^{mmmm}$ or linear combinations of some of them as our computational basis for a $1$-qudit. For a particular theory, this Hilbert space is not big enough for all qudits, but usually we are only interested in a qubit or a qutrit.

To carry out computation, we can braid the first anyon $m$ with the second, the second with the third and the third with the fourth anyon. Each of them corresponds to a unitary transformation on $V_z^{mmmm},$ which we denote by $\sigma_1, \; \sigma_2 $ and $\sigma_3,$ respectively. Moreover, they satisfy the relation:
$$
\sigma_1\sigma_2\sigma_1 = \sigma_2\sigma_1\sigma_2, \; \sigma_2\sigma_3\sigma_2 = \sigma_3\sigma_2\sigma_3, \;\sigma_1\sigma_3 = \sigma_3\sigma_1.
$$

This is just the relation that defines the braid group $\B_4$ on three generators. So we obtain a unitary representation of $\B_4$ on $V_z^{mmmm},$ which we denote by $\rho(m,z)$. The computational power of the theory depends on the image of $\rho(m,z)$ in the special unitary group $\textrm{SU}(V_z^{mmm})$.

\begin{definition}

Given $V_{z}^{mmmm}$, we will call the unitary representation matrices $U(b)$ for braids $b$ the braiding quantum circuits.  The special braiding circuits $U(\sigma_i^{\pm})$ for the braid generators $\sigma_i^{\pm}$ will be called the braiding gates.

The same terminologies are extended to multi-anyons for multi-qudits.

\end{definition}

Since our computational space is always a subspace of the braid group representation, the quantum circuits obtained from braiding quantum circuits are really their restrictions to the computational subspace.  We will not make this distinction when no confusion will arise.

Now we specialize to $D(\Sym_3)$. In order to compute the braiding matrices, we need all the $6j$-symbols and $R$-symbols.  For $D(\Sym_3)$, all of them are listed in Appendix $A$.  To analyze the computational power of $1$-qudit braiding circuits, we need all the representation matrices of $\B_4$.  We systematically analyzed all $\B_4$ representations in Appendix $B$.  These two important appendices are our technical advance.

The natural choice will be to encode a qudit in $V_A^{DDDD}$.  Unfortunately, we did not succeed in finding a model that could be made universal even with measurements and ancillary states.  Therefore, we turn to $V_G^{DDDD}$ based on our knowledge of the braid group representations:

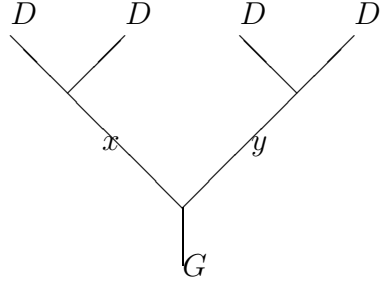
\setlength{\unitlength}{0.030in}
\begin{picture}(160,60)(0,-10)

\put(50,0){\line(0,1){10}}
\put(50,10){\line(1,1){30}}
\put(50,10){\line(-1,1){30}}
\put(70,30){\line(-1,1){10}}
\put(30,30){\line(1,1){10}}

\put(20,42){$D$}
\put(40,42){$D$}
\put(60,42){$D$}
\put(80,42){$D$}
\put(36,20){$x$}
\put(62,20){$y$}
\put(50,-2){$G$}
\end{picture}

The space $V_G^{DDDD}$ is nine dimensional with a basis $\{|GG\rangle, |AG\rangle, |GA\rangle,\\ |FC\rangle, |CF\rangle, |FH\rangle, |HF\rangle, |CH\rangle, |HC\rangle\}$.  Let $U$ = $span\{ |GG\rangle, |AG\rangle, |GA\rangle\},$ $V$ = $span\{\frac{1}{\sqrt{2}}(|FC\rangle + |CF\rangle), \frac{1}{\sqrt{2}}(|FH\rangle + |CH\rangle), \frac{1}{\sqrt{2}}(|HF\rangle + |HC\rangle)\}$ and $W$ = $span\{\frac{1}{\sqrt{2}}(|FC\rangle - |CF\rangle), \frac{1}{\sqrt{2}}(|CH\rangle - |FH\rangle), \frac{1}{\sqrt{2}}(|HF\rangle - |HC\rangle)\}$. To remind ourselves that these basis are used as computational basis, we also write them as $\{|0\rangle_x, |1\rangle_x, |2\rangle_x,\}, \; x = U,\; V,\; W$, where the subscript $x$ indicates which subspace we are referring to, $e.g.$ $|0\rangle_U = |GG\rangle.$ The representation of $\B_4$ splits into the direct sum of a $6$-dim irreducible summand $U \oplus V$  and a $3$-dim irreducible summand $W$.

To encode $2$-qutrits, we consider the following fusion tree:

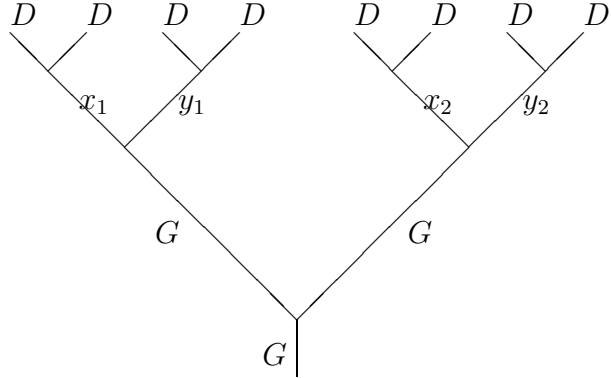
\setlength{\unitlength}{0.020in}
\begin{picture}(180,120)(-40,-50)

\put(50,10){\line(1,-1){45}}
\put(50,10){\line(1,1){30}}
\put(50,10){\line(-1,1){30}}
\put(70,30){\line(-1,1){10}}
\put(30,30){\line(1,1){10}}

\put(140,10){\line(-1,-1){45}}
\put(140,10){\line(1,1){30}}
\put(140,10){\line(-1,1){30}}
\put(160,30){\line(-1,1){10}}
\put(120,30){\line(1,1){10}}

\put(95,-35){\line(0,-1){15}}

\put(20,42){$D$}
\put(40,42){$D$}
\put(60,42){$D$}
\put(80,42){$D$}
\put(38,20){$x_1$}
\put(64,20){$y_1$}
\put(58,-15){$G$}

\put(110,42){$D$}
\put(130,42){$D$}
\put(150,42){$D$}
\put(170,42){$D$}
\put(128,20){$x_2$}
\put(154,20){$y_2$}
\put(124,-15){$G$}

\put(86,-47){$G$}
\end{picture}

The $2$-qutrits are the tensor product of the two qutrits on the two branches.  This encoding of $2$-qutrits is called the {\it sparse encoding} because  encoding with fewer anyons, called the {\it dense encoding}, is also possible.
To encode $n$-qutrits,  we simply use the tensor product of $n$ such branches, so there are totally $4n$ anyons.

We will refer to the three qutrit models that encode $1$-qutrit in the subspaces $U$, $V$, and $W$, respectively, with the computational bases above as the {\it qutrit $U$-model, $V$-model, and $W$-model,} respectively.

To analyze these models, we systematically investigate all relevant braid group representations in Appendix \ref{rep of braid}. Our results are summarized below in Tables $3$ and $4$.

Table \ref{allbasis} list of the dimensions and bases of $V_{z}^{mmmm}$, and Table \ref{irrep} the basic properties of the representations $\rho(m,z).$

{\small
\begin{table}[!h]
\begin{tabular}{|c|c|c|}
\hline
$m,z$     &  dimension  &  basis \\
\hline
$C,A$     &    3        &  $|AA \rangle, |BB \rangle, |CC \rangle$\\
\hline
$C,B$     &    3        &  $|CC \rangle, |AB \rangle, |BA \rangle$\\
\hline
$C,C$     &    5        &  $|CC \rangle, |AC \rangle, |CA \rangle, |BC \rangle, |CB \rangle $\\
\hline
$D,A$     &    5        &  $|AA \rangle, |CC \rangle, |FF \rangle, |GG \rangle, |HH \rangle $\\
\hline
$D,B$     &    4        &  $|CC \rangle, |FF \rangle, |GG \rangle, |HH \rangle $\\
\hline
$D,C$     &    9        &  $|CC \rangle, |AC \rangle, |CA \rangle, |GF \rangle, |FG \rangle, |GH \rangle, |HG \rangle, |FH \rangle, |HF \rangle $\\
\hline
$D,F$     &    9        &  $|FF \rangle, |AF \rangle, |FA \rangle, |GC \rangle, |CG \rangle, |GH \rangle, |HG \rangle, |CH \rangle, |HC \rangle $\\
\hline
$D,G$     &    9        &  $|GG \rangle, |AG \rangle, |GA \rangle, |FC \rangle, |CF \rangle, |FH \rangle, |HF \rangle, |CH \rangle, |HC \rangle $\\
\hline
$D,H$     &    9        &  $|HH \rangle, |AH \rangle, |HA \rangle, |GF \rangle, |FG \rangle, |GC \rangle, |CG \rangle, |FC \rangle, |CF \rangle $\\
\hline
$G,A$     &    3        &  $|AA \rangle, |BB \rangle, |GG \rangle$\\
\hline
$G,B$     &    3        &  $|GG \rangle, |AB \rangle, |BA \rangle$\\
\hline
$G,G$     &    3        &  $|GG \rangle, |AG \rangle, |GA \rangle, |BG \rangle, |GB \rangle$\\
\hline
\end{tabular}
\caption{Dimension and Basis of $V_{z}^{mmmm}$} \label{allbasis}
\end{table}
}

{\small
\begin{table}[!h]
\begin{tabular}{|c|c|c|c|c|c|c|c|}
\hline

\multirow{2}{*}{$m,z$} & Dimension of  & \multicolumn{3}{|l|}{Dimension of} & \multicolumn{3}{|c|}{Images of $\rho(m,z)$ } \\
                     & $V_z^{mmmm}$  & \multicolumn{3}{|l|}{$V_z^{mmmm}$} & \multicolumn{3}{|c|}{on Sectors}             \\
\hline

$C,A$   &  3   &   2 &  1  &   &   $\mathbb{Z}_3 \rtimes \mathbb{Z}_4$  &  1 &  \\
\hline

$C,B$   &  3   &   3 &     &   &   $\Sym_4 $                               &   &  \\
\hline

$C,C$   &  5   &   3 &  1  & 1 &   $\Sym_4$                                &  1 & 1 \\
\hline

$D,A$   &  5   &   3 &  1  & 1 &   $\A_4$                                &  1 & 1 \\
\hline

$D,B$   &  4   &   2 &  2  &   &   $\textrm{SL}(2,\mathbb{F}_3)$                 &  $\textrm{SL}(2,\mathbb{F}_3)$ &  \\
\hline

$D,C$   &   \multicolumn{7}{|c|}{Same as $(D,F)$}\\
\hline

$D,F$   & 9    &   8 &  1  &    &   $\sum(216)$                         &  1 &  \\
\hline

$D,G$   &  9   &   3 &  6  &    &   $\sum(216*3)$                        &  $\sum(216*3)$ &  \\
\hline

$D,H$   &   \multicolumn{7}{|c|}{Same as $(D,G)$}\\
\hline

$G,A$   &  3   &   3 &    &    &   D(9,1,1;2,1,1)                        &       &  \\
\hline

$G,B$   &  3   &   3 &    &    &   D(18,1,1;2,1,1)                        &   &  \\
\hline

$G,G$   &  5   &   4 &  1 &    &    Group of order $648$                      &  1 &  \\
\hline
\end{tabular}
\caption{Summary of the representations $\rho(m,z)$ on $V_{z}^{mmmm}$} \label{irrep}
\end{table}
}

Note that the group of order $648$ in the last row of Table 4 is isomorphic, as an abstract group, to  $(((\mathbb{Z}_3 \times ((\mathbb{Z}_3 \times \mathbb{Z}_3) \rtimes \mathbb{Z}_2)) \rtimes \mathbb{Z}_2 ) \rtimes \mathbb{Z}_3) \rtimes \mathbb{Z}_2$.  This isomorphism is given by the software package GAP.

\subsection{Braiding, measurement, and ancilla}

Using the encoding above, we can simulate standard qutrit quantum circuits by braidings of $D$ anyons.  Concrete braiding quantum circuits are the braid group representation matrices with respect to the fusion tree basis.  For $1$-qutrit braiding circuits, we need to know the representation matrices of $\B_4$---the $4$-strand braid group, and for $2$-qutrit braiding circuits, the representation matrices of $\B_8$.  Since both collections of matrices are finite \cite{ERW}, they are not sufficient to simulate the standard qutrit circuit model.

To gain extra computational power, we consider measurement and ancilla.  In anyon theory, there are two kinds of measurements to determine the total charge of a collection of anyons: projective and interferometric.  Both types of measurements always lead to some decoherence in the model.  Therefore, ideally we should only use them at the end of the computation.  Since we cannot avoid using them for WI anyons, we will allow ourselves to determine whether or not the total charge of two anyons is trivial in the middle of the computation.  Then based on the outcome, we choose how to continue our computation.  For this reason, we call such models {\it adaptive}.

\begin{measurement}\label{Meas 1}
Let $\M_A=\{\Pi_A,\Pi_{A'}\}$ be the projective measurement onto the total charge=$A$ sector and its complement. Then $\M_A$ allows us to
 distinguish between the anyon $A$ and other anyons; namely, check whether a anyon is trivial or not.  Moreover, the state after measurement for each outcome is still coherent.
\end{measurement}

The next measurement that we use is problematic, but it is unavoidable due to our choices of computational subspaces.  It allows us to project states back to the computational subspaces.

\begin{measurement}\label{Meas 2}
Let $S$ be a subspace of an anyonic space and $S^{\perp}$ be its orthonormal complement.  Then $\M_S=\{\Pi_S, \Pi_{S^{\perp}}\}$ is the projective measurement that projects a state to $S$ or $S^{\perp}$.
\end{measurement}

For example, applying $M_S$ to $S=U$ in $V_{G}^{DDDD}$, we obtain the orthogonal projection to $U = span\{ |GG\rangle, |AG\rangle, |GA\rangle\}$ and its orthogonal complement $V \oplus W = span\{ |FC\rangle, |CF\rangle, |FH\rangle, |HF\rangle, |CH\rangle, |HC\rangle\}$.

The main result of the paper is that braiding supplemented by measurements $\M_A$ and $\M_U$ leads to a universal gate set for the $U$-model and $V$-model.  To make the qutrit $W$-model universal, we need to use the extra ancillary state:

\begin{ancilla} \label{Anc 1}
The state in the following picture is denoted by $|H\rangle_A.$
\end{ancilla}
\setlength{\unitlength}{0.030in}
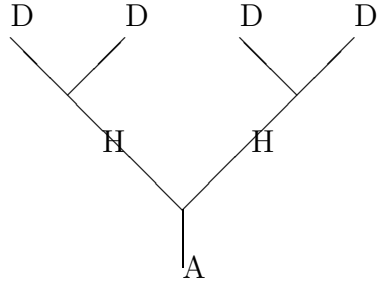
\begin{picture}(160,60)(0,-10)

\put(50,0){\line(0,1){10}}
\put(50,10){\line(1,1){30}}
\put(50,10){\line(-1,1){30}}
\put(70,30){\line(-1,1){10}}
\put(30,30){\line(1,1){10}}

\put(20,42){D}
\put(40,42){D}
\put(60,42){D}
\put(80,42){D}
\put(36,20){H}
\put(62,20){H}
\put(50,-2){A}
\end{picture}

Then our second result is that braiding supplemented by measurements $\M_A$ and $\M_U$ and ancillary state $|H\rangle_A$ leads to a universal gate set for the $W$-model.

\subsection{A universal gate set for qutrits}\label{Hadamardgate}

Theoretically, there is no advantage to use qutrits instead of qubits.  But there are anyon systems that are more natural to choose qutrits rather than qubits.  This is the case when we use WI anyons of quantum dimensions $3$ or $\sqrt{3}$ for anyonic quantum computation.  Moreover, there are some better numbers for qutrits distillation protocols which might provide some benefits for engineering \cite{anwar, vala}.  In this section, we prove that a particular convenient qutrit gate set for our purpose is universal for the standard qutrit circuit model.

The generalized Hadamard gate for qutrit is the following:

$ h = \frac{1}{\sqrt{3}}
\begin{pmatrix}
1       &      1       &     1  \\
1       &     \omega   & \omega^2 \\
1       &     \omega^2   & \omega \\
\end{pmatrix}$,

where $\omega=e^{\frac{2\pi i}{3}}$.

The SUM gate for qudits is a generalized version of $CNOT$, which maps basis element $|i,j\rangle$ to $|i, i+j  \, \textrm{mod}\, 3 \rangle$.
To state our theorem, we need to define another measurement:

\begin{measurement}\label{Meas 3}
Let $\M_{|0\rangle}=\{\Pi_{|0\rangle},\Pi_{{|0\rangle}^{\perp}}\}$ be the projective measurement that is the
orthogonal projection to $span\{|0\rangle\}$ and its orthogonal complement $span\{|1\rangle, |2\rangle\}$ in a qutrit.
\end{measurement}

\begin{thm}\label{HSMA universal}
The $1$-qutrit classical gates, generalized Hadamard gate, $SUM$ gate, and Measurement \ref{Meas 3} form a universal gate set for the standard qutrit quantum circuit model.
\end{thm}

\subsubsection{Proof of Theorem \ref{HSMA universal}} \label{proof universal}
$\\$

We fix $|0\rangle$ and $|1\rangle$ as a qubit, and show that we can implement a universal set of qubit gates.
More explicitly, we use the $2$-dimensional subspace $\mathbb{C}^2 = span\{|0\rangle,|1\rangle\}$ inside $\mathbb{C}^3 = span\{|0\rangle,|1\rangle, |2\rangle\}$ to do computations. During the computations, we will go out of the subspace $\mathbb{C}^2$, and eventually come back to it.  Though unnecessary, we can deduce universality for the qutrit models by enoding a qutrit with two qubits $\mathbb{C}^2 \otimes \mathbb{C}^2 \subset \mathbb{C}^3 \otimes \mathbb{C}^3$. That is, we use $|00\rangle, \, |01\rangle, \, |10\rangle$ to encode $|0\rangle, \, |1\rangle, \, |2\rangle $, respectively. And the basis element $|11\rangle$ is left unused.

Our strategy of proof follows from that of \cite{Kitaev2} and some of the lemmas below are stated in \cite{Kitaev2} as exercises.

Note that with $1$-qutrit classical gates, the generalized Hadamard gate $h$, and Measurement \ref{Meas 3}, we can easily construct the following ancilla and measurements:
$\\$

1). $|i\rangle$, $i = 0, \; 1, \; 2 $.

2). $\widetilde{|i\rangle} = \sum\limits_{j=0}^{2}\omega^{ij}|j\rangle = h|i\rangle$, $i = 0, \; 1, \; 2 $.

3). Projection of a $1$-qutrit state to any computational state, preserving the coherence of the orthogonal complement. For example, projection to $span\{|0\rangle, |1\rangle\}$ and its complement $span\{|2\rangle\}$.

4). Measurement of a qutrit in the standard computational basis.

5). Projection to $span\{\widetilde{|1\rangle}, \widetilde{|2\rangle}\}$ and its complement $span\{\widetilde{|0\rangle}\}$

6). Measurement of a qubit in the standard basis if we take $\{|0\rangle, |1\rangle\}$ as the computational basis. This follows from $4).$
$\\$

%
%
%

From the set of operations given in Theorem \ref{HSMA universal}, we show that we can construct the qutrit (qubit) gates (measurements) in Lemmas \ref{ZSWAPFLIP}, \ref{triple z}, \ref{measure x}.

We define the qutrit gate $FLIP_2$ by the map: $FLIP_2 |0\rangle = |0\rangle, \  FLIP_2 |1\rangle = |1\rangle, \ FLIP_2 |2\rangle = -|2\rangle.$

\begin{lem} \label{ZSWAPFLIP}
The gate $FLIP_2,$ can be constructed.
\begin{proof}

%
%
%
%
%

To obtain $FLIP_2$, we first construct the ancilla $|\psi\rangle = \frac{1}{\sqrt{3}}(|0\rangle - |1\rangle + |2\rangle)$ as follows.

Prepare the state $\widetilde{|1\rangle}\widetilde{|2\rangle},$ and project each qutrit to the space $span\{|0\rangle, |1\rangle\}$ to obtain the state $|\eta\rangle = \frac{1}{2}(|0\rangle + \omega |1\rangle) \otimes (|0\rangle + \omega^2 |1\rangle).$ Apply the $SUM$ gate to $|\eta\rangle$ and then project the first qutrit of the resulting state to the space $span\{\widetilde{|0\rangle}\}$. It's easy to see on the second qutrit we get the state $|\psi\rangle.$

Now for a state $|\phi \rangle = c_0 |0\rangle + c_1|1\rangle + c_2|2\rangle, $ apply the $SUM$ gate to $|\phi\rangle |\psi\rangle$ and then measure the second qutrit in the standard basis. If the outcome is $|0\rangle,$ then the first qutrit is  $c_0 |0\rangle + c_1|1\rangle - c_2|2\rangle. $ If the outcome is $|1\rangle$, then the first qutrit is $-c_0 |0\rangle + c_1|1\rangle + c_2|2\rangle, $ and if the outcome is $|2\rangle,$ then the first qutrit is $c_0 |0\rangle - c_1|1\rangle + c_2|2\rangle $. Moreover, the probability for each case is $\frac{1}{3}$. Therefore, this process changes the sign of some coefficient randomly. By repeating this process, we will get the gate $FLIP_2$.

\end{proof}
\end{lem}


\begin{lem} \label{triple z}
The $3$-qubit gate $\bigwedge^2(\sigma_z)$ which maps $|i,j,k\rangle$ to $(-1)^{ijk}|i,j,k\rangle$ can be constructed. In particular, $\bigwedge(\sigma_z)$ and $\sigma_z$ can be constructed since we have the ancilla $|1\rangle$.
\begin{proof}
Combining the gate $FLIP_2$ obtained in Lemma \ref{ZSWAPFLIP} and the $SUM$ gate, one can construct the following 2-qutrit and 3-qutrit gates.
\begin{equation}
|i,j\rangle \longmapsto
\begin{cases}
-|i,j\rangle  & i + j = 2 \ mod  \ 3 \\
|i,j\rangle   & else \\
\end{cases}
\end{equation}

\begin{equation}
|i,j,k\rangle \longmapsto
\begin{cases}
-|i,j,k\rangle  & i + j + k = 2 \ mod  \ 3 \\
|i,j,k\rangle   & else \\
\end{cases}
\end{equation}

When applying them to the state $|i,j\rangle$ (or $|i,j,k\rangle$), we can describe the above two gates as \lq\lq flip the sign if $i + j = 2 \ mod \ 3$"(or \lq\lq flip the sign if $i+j+k = 2  \ mod \ 3$").

One can check applying the following four gates to a $3$-qubit state $|i,j,k\rangle$ successively gives rise to $\bigwedge^2(\sigma_z).$

 \lq\lq flip the sign if $i+j+k = 2 \ mod \ 3$",

 \lq\lq flip the sign if $i+j = 2 \ mod \ 3$",

 \lq\lq flip the sign if $i+k = 2 \ mod \ 3$",

 \lq\lq flip the sign if $j+k = 2 \ mod \ 3$",

 Note that here $i$, $j$, $k$ are either $0$ or $1$.
\end{proof}
\end{lem}

Let $|\pm\rangle  = \frac{1}{\sqrt{2}}(|0\rangle \pm |1\rangle), $ which are the eigenstates of the qubit gate $\sigma_x.$ Note that the state $|+\rangle$ can be obtained by projecting the state $\widetilde{|0\rangle}$ to the space $span\{|0\rangle, |1\rangle\}.$

\begin{lem} \label{measure x}
Measurement of $\sigma_x$ can be constructed on a qubit.
\begin{proof}
For an arbitrary $1$-qubit state $\alpha|+\rangle + \beta|-\rangle,$ a measurement of $\sigma_x$ would result in the state $|+\rangle$ with probability $|\alpha|^2$, and in the state $|-\rangle$ with probability $|\beta|^2$.

We denote the measurement which projects a state to $span\{\widetilde{|0\rangle}\}$ and its complement $span\{\widetilde{|1\rangle},\widetilde{|2\rangle}\}$ by $\mathcal{M}_1$ and denote the measurement which projects to $span\{|0\rangle, |1\rangle\}$ and its complement $span\{|2\rangle\}$ by $\mathcal{M}_2$.

Note that $|-\rangle$ is orthogonal to $\widetilde{|0\rangle}$ while $|+\rangle$ is not. So if a state results in $\widetilde{|0\rangle}$ after $\mathcal{M}_1$, then the corresponding probability only depends on the $|+\rangle$ component. We explain this idea explicitly below to construct the measurement of $\sigma_x$.

Consider the following procedure.

\begin{equation}
\xymatrix{
\mathcal{O}: \alpha|+\rangle + \beta|-\rangle \ar[r]^-{\mathcal{M}_1} & output \ar[r]^-{Pr = \frac{2|\alpha|^2}{3}} \ar[dl]_-{Pr = 1 - \frac{2|\alpha|^2}{3}} & \widetilde{|0\rangle} \ar[r] & |+\rangle \\
\frac{\alpha|+\rangle + \beta|-\rangle - \frac{2\alpha}{\sqrt{6}}\widetilde{|0\rangle}}{\sqrt{1 - \frac{2|\alpha|^2}{3}}} \ar[r]^-{\mathcal{M}_2} & output \ar[r]_-{Pr = \frac{2|\alpha|^2}{9(1 - \frac{2|\alpha|^2}{3})}} \ar[dl]^-{Pr = 1 - \frac{8|\alpha|^2}{9}} & |2\rangle  \ar[r] & |+\rangle\\
\frac{\alpha|+\rangle + 3\beta|-\rangle}{\sqrt{9-8|\alpha|^2}} & & &\\
}
\end{equation}

So the procedure $\mathcal{O}$ consists of two measurements $\mathcal{M}_1$ and $\mathcal{M}_2$. If the outcome is $\widetilde{|0\rangle}$ after $\mathcal{M}_1$ or is $|2\rangle$ after $\mathcal{M}_2$, then we prepare the state $|+\rangle,$ namely we take the appearance of these two cases as the outcome $|+\rangle$. The probability for either of these two cases to happen is $\frac{8|\alpha|^2}{9}.$ Otherwise, we get the state $\frac{\alpha|+\rangle + 3\beta|-\rangle}{\sqrt{9-8|\alpha|^2}}$ with probability $1 - \frac{8|\alpha|^2}{9}$ and then we iterate the procedure $\mathcal{O}$  until the above two cases happen or the required accuracy is satisfied.

More explicitly, let the resulting state be $|\psi\rangle_n = \alpha_n|+\rangle + \beta_n|-\rangle$ after iterating the procedure $\mathcal{O}$ $n$ times with no $|+\rangle$ outcome and let $b_n$ be the probability to obtain $|\psi\rangle_n$ from $|\psi\rangle_{n-1}$ via the $n$-th procedure. Then we have the following equations.

\begin{equation}
\alpha_n = \frac{\alpha_{n-1}}{\sqrt{9-8|\alpha_{n-1}|^2}}, \ \beta_n = \frac{3\beta_{n-1}}{\sqrt{9-8|\alpha_{n-1}|^2}} , \
b_n = 1 - \frac{8|\alpha_{n-1}|^2}{9} = \frac{|\alpha_{n-1}|^2}{9|\alpha_{n}|^2}
\end{equation}

From the above equations ,we have
\begin{equation}
|\alpha_{n}|^2 = \frac{|\alpha|^2}{(1-|\alpha|^2)9^n + |\alpha|^2}, \ \prod\limits_{i=1}^n b_i =  |\beta|^2 + \frac{|\alpha|^2}{9^n}
\end{equation}

So the probability for iterating the procedure $n$ times with no $|+\rangle$ outcome is $b = \prod\limits_{i=1}^n b_i =  |\beta|^2 + \frac{|\alpha|^2}{9^n}, $ which is very close to $|\beta|^2$ when $n$ is large. Moreover, $|\langle + | \psi_n \rangle|^2 = |\alpha_{n}|^2$ which is close to zero, namely $\psi_n$ is almost equal to $|-\rangle$ up to a phase. Therefore, it's reasonable to treat the case that no $|+\rangle$ appears within $n$ procedures for some proper large $n$, as the outcome $|-\rangle$.

To sum up, after iterating the procedure $n$ times, we can get the state $|+\rangle$ with probability $1 - b = (1-\frac{1}{9^n})|\alpha|^2$ and $|-\rangle$ with probability $b = |\beta|^2 + \frac{|\alpha|^2}{9^n}.$ If we take $n$ large enough,  we get the measurement of $\sigma_x$ with required accuracy.
\end{proof}
\end{lem}

\begin{lem}  \label{Toffoli universal} \cite{Mochon}
The following set of qubit operations are universal for quantum computation:

1). Create the state $|\pm\rangle = \frac{1}{\sqrt{2}}(|0\rangle \pm |1\rangle), |0\rangle$ and $|1\rangle$

2). Measure $\sigma_z$.

3). Measure $\sigma_x$.

4). The Toffoli gate $T = \bigwedge^2(\sigma_x).$
\end{lem}

For a proof of this lemma, see \cite{Mochon}.

\begin{lem} \label{Double Z}  \cite{Kitaev2}
The following set of qubit operations are universal for quantum computation:

1). Create the state $|+\rangle = \frac{1}{\sqrt{2}}(|0\rangle + |1\rangle).$

2). Measure $\sigma_z$.

3). Measure $\sigma_x$.

4). The gate $\bigwedge^2(\sigma_z).$
\end{lem}

\begin{proof}
We prove this lemma by showing that the set of operations here can be used to implement all the operations in Lemma \ref{Toffoli universal}.

Since we can measure $\sigma_x$ and $\sigma_z,$ it's clear that $|-\rangle,\; |0\rangle$ and $|1\rangle$ all can be created from $|+\rangle.$ Thus, it suffices to show the Toffoli gate can be created.

For notational convenience, we also denote $|\widetilde{0}\rangle = |+\rangle,\; |\widetilde{1}\rangle = |-\rangle$ in the following proof. The readers shouldn't be confused with this notation and the one we used for a qutrit, since for the moment we only work in qubit space.

 With the ancilla $|1\rangle$ and the gate $\bigwedge^{2}(\sigma_z),$ we can get the gates $\bigwedge(\sigma_z)$ and $\sigma_z$.

 Next we do the following procedure which creates a \lq\lq gate" $H$ or $\sigma_x H$,

 $|i\rangle|\widetilde{0}\rangle \xlongrightarrow{\quad \bigwedge(\sigma_z) \quad} |i\rangle|\widetilde{i}\rangle \xlongrightarrow{\quad \textrm{Measure} (\sigma_x)_1 \quad}$
 $\begin{cases}
 |\widetilde{0}\rangle|\widetilde{i}\rangle & {\textrm{outcome} \quad \textrm{is} \quad 1}\\
 |\widetilde{1}\rangle(-1)^{i}|\widetilde{i}\rangle & {\textrm{outcome} \quad \textrm{is} \quad -1}\\
 \end{cases} $

 By measuring $(\sigma_x)_{1},$ we mean measuring $\sigma_x$ on the first qubit.

 One checks that both probabilities are $\frac{1}{2}.$ If the outcome is $-1$, then we continue to apply the gate $\sigma_z$ on the first qubit so that the state becomes $|\widetilde{0}\rangle(-1)^{i}|\widetilde{i}\rangle$.

 Notice that our ancilla starts from the second qubit while ends on the first qubit, $i.e.$ the working qubit and ancilla qubit are switched. But we show below that this is not a problem.

 Therefore, if the outcome is $1$, we produced the gate qubit Hadamard gate $H$, and otherwise we produced $\sigma_x H.$ We name this sequence of operations by $\mathcal{A}$.

 Now we produce the gate $T$.

 $|i,j,k\rangle|\widetilde{0}\rangle \xlongrightarrow{\quad \mathcal{A}_{3,4} \quad}
 \begin{cases}
 |i,j,\widetilde{0},\widetilde{k}\rangle & \textrm{outcome} \quad 1\\
 (-1)^{k}|i,j,\widetilde{0},\widetilde{k}\rangle & \textrm{outcome} \quad -1 \\
 \end{cases}
 \quad \xlongrightarrow{\quad \bigwedge^{2}(\sigma_z)_{1,2,4} \quad}
 \begin{cases}
 |i,j,\widetilde{0},\widetilde{ij+k}\rangle & \\
 (-1)^{k}|i,j,\widetilde{0},\widetilde{ij+k}\rangle &  \\
 \end{cases}
 \\
 \quad \xlongrightarrow{\quad \mathcal{A}_{4,3} \quad}
 \begin{cases}
 |i,j,ij+k,\widetilde{0}\rangle & \textrm{outcome} \quad (1,1)\\
 (-1)^{k}|i,j,ij+k,\widetilde{0}\rangle & \textrm{outcome} \quad (-1,1)\\
 |i,j,ij+k+1,\widetilde{0}\rangle & \textrm{outcome} \quad (1,-1)\\
 (-1)^{k}|i,j,ij+k+1,\widetilde{0}\rangle & \textrm{outcome} \quad (-1,-1)\\
 \end{cases}$

 In the diagram above, $\mathcal{A}_{3,4} $ means applying the operation $\mathcal{A}$ with the third qubit as working bit and the fourth qubit as ancilla. Each pair of outcome happens with probability $\frac{1}{4}.$

 If the outcome is (1,1) or (1,-1), we do nothing.

 If the outcome is (-1,1) or (-1,-1), we apply the gate $(\sigma_z)_3 \bigwedge(\sigma_z)_{1,2}$ to get the state $|i,j,ij+k\rangle$ or $-|i,j,ij+k+1,\widetilde{0}\rangle$. The overall phase is not important.

 Therefore, if the outcome is (1,1) or (-1,1), we produced the gate $T$. Otherwise we got the gate $(\sigma_x)_3 T$. Both probabilities are $\frac{1}{2}$.

 In the latter case, we repeat the procedure, then we either go back to the original state with probability $\frac{1}{2}$, or we go to the state $|i,j,k+1\rangle$ also with probability $\frac{1}{2}.$ Repeat the procedure again. It's easy to see that after doing this procedure at most 3 times. the probability to get the state $|i,j,ij+k\rangle$ is $\frac{1}{2} + \frac{1}{2^2}$.

 After at most $2n-1$ times, the probability to get $T$ is $\frac{1}{2} + \frac{1}{2^2} + \cdots + \frac{1}{2^{n}} = 1- \frac{1}{2^n}.$

 Therefore, after repeating enough times, we will eventually produce the gate $T$.
\end{proof}

By the lemmas above in this subsection, all the operations in Lemma \ref{Double Z} can be created from the operations given in Theorem \ref{HSMA universal} if we pick a qubit from the qutrit space. Thus Lemma \ref{Double Z} implies Theorem \ref{HSMA universal}.

\section{Universal adaptive anyonic computing models} \label{universal}

In this section we prove that the $U$-model, $V$-model, and $W$-model in Section \ref{UVWmodel} can be made universal provided measurement and ancilla are allowed besides braiding.

Recall that
 $U$ = $span\{ |GG\rangle, |AG\rangle, |GA\rangle\},$ $V$ = $span\{\frac{1}{\sqrt{2}}(|FC\rangle + |CF\rangle), \frac{1}{\sqrt{2}}(|FH\rangle + |CH\rangle), \frac{1}{\sqrt{2}}(|HF\rangle + |HC\rangle)\}$ and $W$ = $span\{\frac{1}{\sqrt{2}}(|FC\rangle - |CF\rangle), \frac{1}{\sqrt{2}}(|CH\rangle - |FH\rangle), \frac{1}{\sqrt{2}}(|HF\rangle - |HC\rangle)\}$.  The computational basis for the three models are denoted as $\{|0\rangle_x, |1\rangle_x, |2\rangle_x,\}$ corresponding to the pair of anyons above, $x = U,\; V,\; W$, where the subscript $x$ indicates which subspace we are referring to.

 Our main theorems are:

\begin{thm} \label{UV universal}
Braiding quantum gates and  Measurements \ref{Meas 1} and \ref{Meas 2} provide a universal gate set for the qutrit $U$-model and $V$-model.
\end{thm}

\begin{thm}\label{W universal}

A universal gate set for the $W$-model can be constructed from braidings and Measurements \ref{Meas 1}, \ref{Meas 2} when the ancillary state \ref{Anc 1} is used.
\end{thm}

The proof of Theorem \ref{UV universal} is given in the next subsection and the proof of Theorem \ref{W universal} in Section \ref{proofWuniversal}.

\subsection{Universality for $U$- and $V$-models}

$U \oplus V$ is a $6$-dim irreducible representation of $\B_4$. Under the basis $span\{ |0\rangle_U, |1\rangle_U, |2\rangle_U, |0\rangle_V,\\ |1\rangle_V, |2\rangle_V\},$ the generators $\sigma_i \, 's$ have the following matrices:

$\sigma_1 =
\begin{pmatrix}
\omega^2 & 0 & 0 & 0 & 0 & 0 \\
 0       & 1 & 0 & 0 & 0 & 0 \\
 0       & 0 & \omega^2 & 0 & 0 & 0 \\
 0       & 0 & 0 & 1 & 0 & 0 \\
 0       & 0 & 0 & 0 & 1 & 0 \\
 0       & 0 & 0 & 0 & 0 & \omega \\
\end{pmatrix}$

$\sigma_2 = \frac{1}{3}
\begin{pmatrix}
1                &       \omega &  \omega         &  \sqrt{2}\omega^2    &   \sqrt{2}       &  \sqrt{2}  \\
\omega           &          1   &  \omega         &  \sqrt{2}            &   \sqrt{2}       &   \sqrt{2}\omega^2\\
\omega           &      \omega  &   1             &   \sqrt{2}           & \sqrt{2}\omega^2 &  \sqrt{2}\\
\sqrt{2}\omega^2 &      \sqrt{2}&   \sqrt{2}      &  -\omega             &  -\omega^2       &  -\omega^2  \\
 \sqrt{2}        &     \sqrt{2} &\sqrt{2}\omega^2 &   -\omega^2          &  -\omega         &  -\omega^2 \\
 \sqrt{2}        &\sqrt{2}\omega^2 & \sqrt{2}     &        -\omega^2     &  -\omega^2       & -\omega\\
\end{pmatrix}
$

$\sigma_3 =
\begin{pmatrix}
\omega^2 & 0 & 0 & 0 & 0 & 0 \\
 0       & \omega^2 & 0 & 0 & 0 & 0 \\
 0       & 0 & 1 & 0 & 0 & 0 \\
 0       & 0 & 0 & 1 & 0 & 0 \\
 0       & 0 & 0 & 0 & \omega & 0 \\
 0       & 0 & 0 & 0 & 0 & 1 \\
\end{pmatrix}$

Let $p = \sigma_1\sigma_2\sigma_1$ and $q = \sigma_3\sigma_2\sigma_3$. Then

$p^2 =
\begin{pmatrix}
0 & 0 & 1 & 0 & 0 & 0 \\
0 & 1 & 0 & 0 & 0 & 0 \\
1 & 0 & 0 & 0 & 0 & 0 \\
0 & 0 & 0 & 0 & 1 & 0 \\
0 & 0 & 0 & 1 & 0 & 0 \\
0 & 0 & 0 & 0 & 0 & 1 \\
\end{pmatrix}$

$q^2 =
\begin{pmatrix}
0 & 1 & 0 & 0 & 0 & 0 \\
1 & 0 & 0 & 0 & 0 & 0 \\
0 & 0 & 1 & 0 & 0 & 0 \\
0 & 0 & 0 & 0 & 0 & 1 \\
0 & 0 & 0 & 0 & 1 & 0 \\
0 & 0 & 0 & 1 & 0 & 0 \\
\end{pmatrix}$

$p^2q^2p^2 =
\begin{pmatrix}
1 & 0 & 0 & 0 & 0 & 0 \\
0 & 0 & 1 & 0 & 0 & 0 \\
0 & 1 & 0 & 0 & 0 & 0 \\
0 & 0 & 0 & 1 & 0 & 0 \\
0 & 0 & 0 & 0 & 0 & 1 \\
0 & 0 & 0 & 0 & 1 & 0 \\
\end{pmatrix}$

Therefore, when restricted to the subspace $U$ or $V$, $p^2$ and $q^2$ generate all the classical gates on $1$ qutrit and $p^2q^2p^2$ is equal to $h^2$, where $h$ is the generalized Hadamard gate defined in Section \ref{Hadamardgate}.

Let $h' = q^2 p q^2$. Then

$h' = \frac{1}{\sqrt{3}}
\begin{pmatrix}
h               &  \sqrt{2}h^{-1} \\
\sqrt{2}h^{-1}  &  -h              \\
\end{pmatrix}
\qquad
h'^{-1} = \frac{1}{\sqrt{3}}
\begin{pmatrix}
h^{-1}               &  \sqrt{2}h \\
\sqrt{2}h            &   -h^{-1}              \\
\end{pmatrix}
$

Define a unitary transformation $\gamma: U \longrightarrow V, $ $\gamma |j\rangle_U = |j\rangle_V, \;j = 0,\; 1, \; 2.$

\begin{lem}\label{hadamard}
By alternating use of $h'$ (or $h'^{-1}$) and Measurement \ref{Meas 2}, one can eventually obtain the generalized Hadamard gate on both $U$ and $V$, as well as the transformations $\gamma$ and $\gamma^{-1}.$ Moreover, the probability to successfully construct these transformations approaches to 1 exponentially fast in the number of measurements and the gate $h'$.
\begin{proof}
Let the generalized Hadamard gate $h$ act on both the spaces $U$ and $V$, then we have
$$h'|j\rangle_U  = \frac{1}{\sqrt{3}}(h|j\rangle_U + \sqrt{2}h^{-1}|j\rangle_V)$$
and
$$h'|j\rangle_V  = \frac{1}{\sqrt{3}}(\sqrt{2}h^{-1}|j\rangle_U - h|j\rangle_V)$$

We first construct the Hadamard gate on $U$ first.

Denote the operation of Measurement \ref{Meas 2} by $\mathcal{M}$. Consider the following procedures:

\begin{equation}
\xymatrix{
\mathcal{P}:|j\rangle_U \ar[r]^-{h'} & \frac{1}{\sqrt{3}}(h|j\rangle_U + \sqrt{2}h^{-1}|j\rangle_V)\ar[r]^-{\mathcal{M}} & output \ar[r]^-{Pr = \frac{1}{3}} \ar[dll]_-{Pr = \frac{2}{3}} & h|j\rangle_U & \\
h^{-1}|j\rangle_V \ar[r]^-{h'} & \frac{1}{\sqrt{3}}(\sqrt{2}h^2|j\rangle_U - |j\rangle_V) \ar[r]^-{\mathcal{M}} & output \ar[r]^-{Pr = \frac{2}{3}} \ar[dll]^-{Pr = \frac{1}{3}} & h^2|j\rangle_U \ar[d]_-{p^2q^2p^2} & \\
|j\rangle_V  & & & |j\rangle_U& \\
}
\end{equation}

\begin{equation}
\xymatrix{
\mathcal{Q}:|j\rangle_V \ar[r]^-{h'^{-1}} & \frac{1}{\sqrt{3}}(\sqrt{2}h|j\rangle_U - h^{-1}|j\rangle_V)\ar[r]^-{\mathcal{M}} & output \ar[r]^-{Pr = \frac{2}{3}} \ar[dll]_-{Pr = \frac{1}{3}} & h|j\rangle_U & \\
-h^{-1}|j\rangle_V \ar[r]^-{h'^{-1}} & \frac{1}{\sqrt{3}}(-\sqrt{2}|j\rangle_U +  h^2|j\rangle_V) \ar[r]^-{\mathcal{M}} & output \ar[r]^-{Pr = \frac{2}{3}} \ar[dll]^-{Pr = \frac{1}{3}} & |j\rangle_U & \\
h^2|j\rangle_V \ar[r]^-{p^2q^2p^2}  & |j\rangle_V& & &\\
}
\end{equation}

Thus, if we run the procedure $\mathcal{P}$ on the space $U$, we have a probability $\frac{1}{3}$ to obtain the Hadamard $h$, $\frac{4}{9}$ to obtain the identity and $\frac{2}{9}$ to obtain the transformation $\gamma.$ If we obtained $h$, then we are done. If we constructed the identity gate, then we run the procedure $\mathcal{P}$ again. If we got the transformation $\gamma,$ then we apply the procedure $\mathcal{Q}$ to the resulting state. After running $\mathcal{Q},$ we have a probability $\frac{2}{3}$ to obtain $h$, $\frac{2}{9}$ to go back to the original state $|j\rangle_U,$ and $\frac{1}{9}$ to get the state $|j\rangle_V.$ Repeat the procedures $\mathcal{P}$ and$/$or $\mathcal{Q}$, according to which space the resulting state after each procedure is in, until we get the Hadamard gate $h$. And it's not hard to show that the probability to construct $h$ within $n$ procedures is $1-\frac{2}{3}\cdot(\frac{5}{9})^{n-1}$, which approaches to $1$ exponentially fast.

The Hadamard gate on $V$ can be constructed in the same way.

To construct the transformation $\gamma,$ see the following procedure:

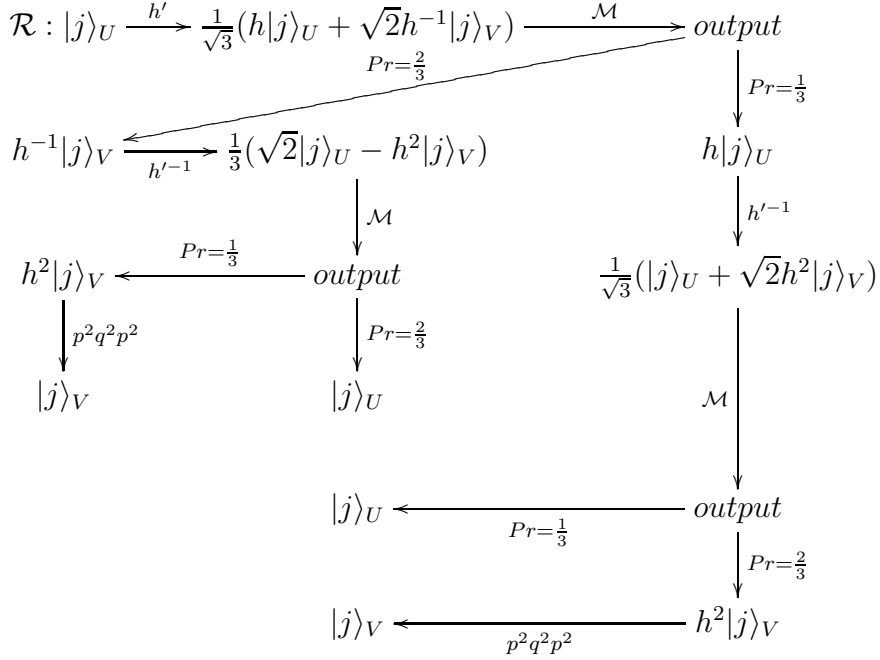
\begin{equation}
\xymatrix{
\mathcal{R}: |j\rangle_U \ar[r]^-{h'} & \frac{1}{\sqrt{3}}(h|j\rangle_U + \sqrt{2}h^{-1}|j\rangle_V)\ar[r]^-{\mathcal{M}} & output \ar[d]^-{Pr = \frac{1}{3}} \ar[dll]_-{Pr = \frac{2}{3}} \\
h^{-1}|j\rangle_V \ar[r]_-{h'^{-1}} &  \frac{1}{3}(\sqrt{2}|j\rangle_U - h^2|j\rangle_V) \ar[d]^-{\mathcal{M}}  &  h|j\rangle_U \ar[d]^-{h'^{-1}} \\
h^2|j\rangle_V \ar[d]^-{p^2q^2p^2} & output \ar[d]^-{Pr = \frac{2}{3}} \ar[l]_-{Pr = \frac{1}{3}}& \frac{1}{\sqrt{3}}(|j\rangle_U + \sqrt{2}h^2|j\rangle_V) \ar[dd]_-{\mathcal{M}}\\
|j\rangle_V &  |j\rangle_U &   \\
&|j\rangle_U & output \ar[l]^-{Pr = \frac{1}{3}} \ar[d]^-{ Pr = \frac{2}{3}} \\
 & |j\rangle_V  & h^2|j\rangle_V \ar[l]^-{p^2q^2p^2} \\
}
\end{equation}

So the procedure $\mathcal{R}$ has a probability of $\frac{4}{9}$ to construct the transformation $\gamma,$ and a probability of $\frac{5}{9}$ to obtain the identity. By repeating it, one can show the probability to construct $\gamma$ within $n$ times is $1 - (\frac{5}{9})^n.$ Therefore, one can obtain $\gamma$ exponentially fast.

Similarly, one can construct $\gamma^{-1}.$
\end{proof}
\end{lem}

The following lemma shows Measurement \ref{Meas 3} can be constructed in both $U$ and $V$.

\begin{lem}
Using Measurement \ref{Meas 1}, \ref{Meas 2} and braiding, one can perform Measurement \ref{Meas 3} in both the space $U$ and $V$.
\begin{proof}
Note that we used the notation $|0\rangle_U = |GG\rangle, \; |1\rangle_U = |AG\rangle, \; \\ |2\rangle_U = |GA\rangle$. Given a state $|\psi\rangle = a |GG\rangle + b|AG\rangle + c|GA\rangle$ in $U$, we apply Measurement \ref{Meas 1} to the left half of the state, $i.e$, we check whether or not the first pair of $D$ anyons in the $1$-qudit splitting tree has total trivial charge. This is essentially the projection to $span\{|AG\rangle\}$ and its orthogonal complement in $U$, namely the projection to $span\{|1\rangle_U\}$ and $span\{|0\rangle_U, |2\rangle_U\}$. Since we have all the $1$-qutrit classical gates on $U$, it's clear that Measurement \ref{Meas 3} in $U$ can be constructed.

Measurement \ref{Meas 3} in $V$ follows from Lemma \ref{hadamard} that one can construct the transformation $\gamma,\; \gamma^{-1}$ to go back and forth between $U$ and $V$.
\end{proof}
\end{lem}

Up to now, we only considered gates and operations on one qutrit. Next, we want to construct a $2$-qutrit gate, the Controlled-$Z$ gate $\bigwedge(Z)$ which maps $|i,j\rangle$ to $\omega^{ij}|i,j\rangle.$

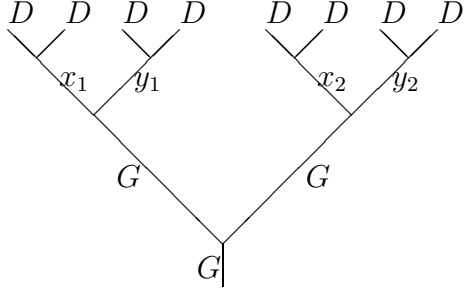
\setlength{\unitlength}{0.015in}
\begin{picture}(180,120)(-40,-50)

\put(50,10){\line(1,-1){45}}
\put(50,10){\line(1,1){30}}
\put(50,10){\line(-1,1){30}}
\put(70,30){\line(-1,1){10}}
\put(30,30){\line(1,1){10}}

\put(140,10){\line(-1,-1){45}}
\put(140,10){\line(1,1){30}}
\put(140,10){\line(-1,1){30}}
\put(160,30){\line(-1,1){10}}
\put(120,30){\line(1,1){10}}

\put(95,-35){\line(0,-1){15}}

\put(20,42){$D$}
\put(40,42){$D$}
\put(60,42){$D$}
\put(80,42){$D$}
\put(38,20){$x_1$}
\put(64,20){$y_1$}
\put(58,-15){$G$}

\put(110,42){$D$}
\put(130,42){$D$}
\put(150,42){$D$}
\put(170,42){$D$}
\put(128,20){$x_2$}
\put(154,20){$y_2$}
\put(124,-15){$G$}

\put(86,-47){$G$}
\end{picture}

We use the above fusion tree to encode $2$-qutrits. Let $s_1 = \sigma_2\sigma_1\sigma_3\sigma_2,$ namely, $s_1$ is the braiding of the first pair with the second pair. Similarly let $\ s_2 = \sigma_4\sigma_3\sigma_5\sigma_4,\ s_3 = \sigma_6\sigma_5\sigma_7\sigma_6.\ $ Clearly $s_1$ exchanges $x_1$ with $y_1$ with a phase in the above $2$-qudit splitting tree, namely it maps $|x_1,y_1;x_2,y_2 \rangle$ to $|y_1,x_1;x_2,y_2 \rangle$ up to a phase. Similarly, $s_3$ exchanges $x_2$ with $y_2$. The gate $s_2$ is much more complicated since it involves $F$-moves. Let $CrlZ = s_1^{-1}s_2^2s_1 s_3^{-1}s_2^2 s_3.$ Through direct calculations, we found $CrlZ$ is a diagonal matrix. Moreover, when restricted to the space $U$, $CrlZ$ is exactly the Controlled-$Z$ gate $\bigwedge(Z).$ Again, via the transformation $\gamma,$ one also obtains the Controlled-$Z$ gate in the space $V$.

The $SUM$ gate maps $|i,j\rangle$ to $|i, i+j \rangle$ and can be obtained by conjugating $\bigwedge(Z)$ via the Hadamard. Explicitly,
$$SUM = (Id \otimes h)\bigwedge(Z)^{-1}(Id \otimes h^{-1}).$$

So we can also construct the $SUM$ gate in the space $U$ and $V$.

~\\
To sum up, with Measurement \ref{Meas 1}, \ref{Meas 2} and braiding, we can construct all the 1-qutrit classical gates, generalized Hadamard gate, $SUM$ gate and Measurement \ref{Meas 3} in both the space $U$ and $V$.

Finally, Theorem \ref{UV universal} follows from Theorem \ref{HSMA universal} and the arguments in this subsection.

\subsection{Universality for $W$-model}\label{proofWuniversal}

In this subsection, we examine the representation on $W$. Under the basis of $W$ given by $\{ |0\rangle_W, |1\rangle_W, |2\rangle_W\}$, the $\sigma_i \, 's$ have the matrices:

$\sigma_1 =
\begin{pmatrix}
1 &  0  &  0 \\
0 &  1 &  0 \\
0 &  0  &  \omega \\
\end{pmatrix}
\qquad
\sigma_3 =
\begin{pmatrix}
1 &  0  &  0 \\
0 &  \omega &  0 \\
0 &  0  &  1 \\
\end{pmatrix}
$

$\sigma_2 =
\begin{pmatrix}
\frac{1}{2} + \frac{\sqrt{3}i}{6} &  -\frac{1}{2} + \frac{\sqrt{3}i}{6}  &  -\frac{1}{2} + \frac{\sqrt{3}i}{6} \\
-\frac{1}{2} + \frac{\sqrt{3}i}{6} & \frac{1}{2} + \frac{\sqrt{3}i}{6}   &  -\frac{1}{2} + \frac{\sqrt{3}i}{6} \\
-\frac{1}{2} + \frac{\sqrt{3}i}{6} &  -\frac{1}{2} + \frac{\sqrt{3}i}{6}  &  \frac{1}{2} + \frac{\sqrt{3}i}{6} \\
\end{pmatrix}
$

The same as last subsection, define $p$ = $\sigma_1\sigma_2\sigma_1$, $q$ = $\sigma_3\sigma_2\sigma_3$. Then

$p^2 =-
\begin{pmatrix}
0 &  1  &  0 \\
1 &  0 &  0 \\
0 &  0  &  1 \\
\end{pmatrix}
\qquad
q^2 =-
\begin{pmatrix}
0 &  0  &  1 \\
0 &  1 &  0 \\
1 &  0  &  0 \\
\end{pmatrix}
$

So $p^2$ and $q^2$ generate all the $1$-qutrit classical gates in $W$.

Also from $\sigma_1$ and $\sigma_3$, we obtain the generalized $Z$-gate and Phase gate $P$:

$
Z =
\begin{pmatrix}
1 &  0  &  0 \\
0 &  \omega &  0 \\
0 &  0  &  \omega^2 \\
\end{pmatrix}
\qquad
P =
\begin{pmatrix}
1 &  0  &  0 \\
0 &  1 &  0 \\
0 &  0  &  \omega \\
\end{pmatrix}
$ where $Z$ maps $|i\rangle$ to $\omega^i|i\rangle$ and $P$ maps $|i\rangle$ to $\omega^{\frac{i^2-i}{2}}|i\rangle.$

Moreover, let $h' = q^2pq^2$, then
$
h' = \frac{1}{\sqrt{3}i}
\begin{pmatrix}
1 &  1  &  1 \\
1 &  \omega & \omega^2 \\
1 &  \omega^2  &  \omega \\
\end{pmatrix}
$, which is exactly the generalized Hadamard gate up to a phase.

Therefore, in the space $W$, we obtained the classical 1-qutrit gates, generalized $Z$-gate, the Phase gate and the generalized Hadamard gate by braiding.

Now we turn to constructing the $2$-qutrit gate $\bigwedge(Z).$ One may try the same braiding method as we did for the space $U$. But it turns out that braiding doesn't work for $W$. Instead, we try to construct a transformation  similar to $\gamma$.

Consider the following picture of braiding.

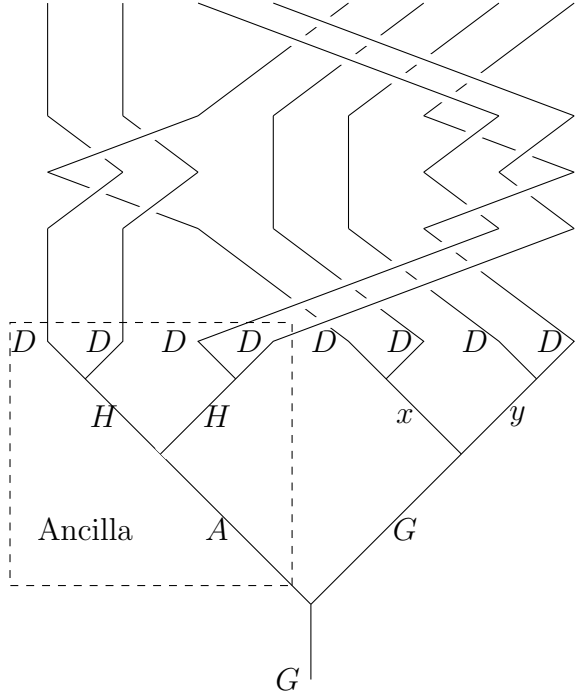
\begin{tikzpicture}[scale = 0.5]

       \begin{scope}
       \draw (0,0) -- (0,-2) node[left]{$G$};
        \draw (4,4) -- (1,7) node[left]{$D$} -- (-3,10);
       \draw (2,6) -- (3,7)node[left]{$D$} -- (-1,10);
       \draw (6,6) -- (5,7) node[left]{$D$}-- (1,10);
       \draw (0,0) -- (7,7) node[left]{$D$}-- (3,10);

       \draw (0,0) -- (-7,7) node[left]{$D$}-- (-7,10);
       \draw (-6,6) -- (-5,7) node[left]{$D$}-- (-5,10);
       \draw (-2,6) -- (-3,7)[color = white, line width = 2mm] -- (5,10);
       \draw (-4,4) -- (-1,7)[color = white, line width = 2mm] -- (7,10);
       \draw (-2,6) -- (-3,7)node[left]{$D$} -- (5,10);
       \draw (-4,4) -- (-1,7)node[left]{$D$} -- (7,10);

       \draw (-3,10)--(-7,11.5);
       \draw [color = white, line width = 2mm](-6.6,10.3)-- (-5,11.5);
       \draw [color = white, line width = 2mm](-4.6,10.3) -- (-3,11.5);
       \draw (-7,10)-- (-5,11.5);
       \draw (-5,10) -- (-3,11.5);
       \draw (5,10)-- (3,11.5);
       \draw (7,10)-- (5,11.5);
       \draw [color = white, line width = 2mm](3.8,10.3)-- (7,11.5);
       \draw (3,10)-- (7,11.5);
       \draw (-1,10)-- (-1,13);
       \draw (1,10) -- (1,13);

       \draw (-5,11.5) -- (-7,13);
       \draw (-3,11.5) -- (-5,13);
       \draw [color = white, line width = 2mm](-6.2,11.8) -- (-3,13);
       \draw (-7,11.5) -- (-3,13);
       \draw (7,11.5) -- (3,13);
       \draw [color = white, line width = 2mm](3.4,11.8) -- (5,13);
       \draw [color = white, line width = 2mm](5.4,11.8) -- (7,13);
       \draw (3,11.5) -- (5,13);
       \draw (5,11.5) -- (7,13);

       \draw (-7,13)-- (-7,16);
       \draw (-5,13)-- (-5,16);
       \draw (-3,13)-- (1,16);
       \draw (-1,13)-- (3,16);
       \draw (1,13)-- (5,16);
       \draw (3,13)-- (7,16);
       \draw [color = white, line width = 2mm](4.2,13.3)-- (-3,16);
       \draw [color = white, line width = 2mm](6.2,13.3)-- (-1,16);
       \draw (5,13)-- (-3,16);
       \draw (7,13)-- (-1,16);

       \draw (-5.5,5) node{$H$};
       \draw (-2.5,5) node{$H$};
       \draw (2.5,5) node{$x$};
       \draw (5.5,5) node{$y$};

       \draw (-2.5,2) node{$A$};
       \draw (2.5,2) node{$G$};

      \draw[dashed] (-0.5,0.5)--(-0.5,7.5) -- (-8,7.5) -- (-8,0.5) -- (-0.5,0.5);
      \draw (-6,2) node{$\textrm{Ancilla}$};
       \end{scope}
\end{tikzpicture}

Let $P$= $\sigma_6\sigma_5\sigma_4\sigma_3\sigma_7\sigma_6\sigma_5\sigma_4$, $Q$ = $\sigma_2\sigma_1\sigma_1\sigma_2\sigma_6\sigma_7\sigma_7\sigma_6$, and let $R$ = $P^{-1}QP$. Then the braiding in the picture is given by $R$.

We denote the state in the picture before braiding by $|H\rangle_A |xy\rangle$. Then the braiding $R$ gives the following transformation:

$$|H\rangle_A |i\rangle_W \longmapsto \frac{1}{2}(-|H\rangle_A |i\rangle_W + |H\rangle_B |i\rangle_V - \sqrt{2}|H\rangle_B |-i\rangle_U)$$

and

$$|H\rangle_B |i\rangle_U \longmapsto \frac{1}{\sqrt{2}}(|H\rangle_A |-i\rangle_W + |H\rangle_B |-i\rangle_V )$$

where $i = 0, \; 1, \; 2$ and $-i$ is taken to be modulo $3$.

Define a unitary transformation $\beta: |H\rangle_A \otimes  W \longrightarrow |H\rangle_B \otimes U, \, \\ \beta(|H\rangle_A|i\rangle_W) = |H\rangle_B|i\rangle_U.$ Here $|H\rangle_A$ is the ancilla.

\begin{lem}
With braiding, Measurement \ref{Meas 1}, \ref{Meas 2} and Ancilla \ref{Anc 1}, the transformation $\beta$ and $\beta^{-1}$ can be constructed with probability approaching to 1 exponentially fast in the number of measurements and the gates applied.

\begin{proof}
In the following diagram, $\mathcal{M}_1$ means applying Measurement \ref{Meas 1} to the first qudit (the ancilla part) to check whether the total charge is trivial or not and $\mathcal{M}_2$ is Measurement \ref{Meas 2} applied to the second qudit. Consider the following procedure $\mathcal{S}$:

\begin{equation}
 \xymatrix{
 \mathcal{S}:|H\rangle_A|i\rangle_W \ar[d]^-{R}   &   \\
 \frac{1}{2}(-|H\rangle_A |i\rangle_W + |H\rangle_B |i\rangle_V - \sqrt{2}|H\rangle_B |-i\rangle_U) \ar[r]^-{\mathcal{M}_1}  &output \ar[d]_-{Pr = \frac{1}{4}} \ar[dl]^-{Pr = \frac{3}{4}}  \\
 \frac{1}{\sqrt{3}}(|H\rangle_B|i\rangle_V - \sqrt{2}|H\rangle_B|-i\rangle_U) \ar[d]^-{\mathcal{M}_2}     & |H\rangle_A|i\rangle_W \\
 output \ar[dr]^-{Pr = \frac{1}{3}} \ar[d]_-{Pr = \frac{2}{3}}  &    \\
 |H\rangle_B|-i\rangle_U \ar[d]^-{Id \otimes p^2q^2p^2}  & |H\rangle_B|i\rangle_V \ar[d]^-{Id \otimes \gamma^{-1}} \\
 |H\rangle_B|i\rangle_U  & |H\rangle_B|i\rangle_U
}
\end{equation}

Starting from the state $|H\rangle_A|i\rangle_W$ with $|H\rangle_A$ as ancilla, we apply the procedure $\mathcal{S}$ to it. From the diagram above, one can see that there is a probability of $\frac{1}{4}$ for the state to remain unchanged in which case we would apply the procedure again. Otherwise, the state is transformed to $|H\rangle_B|i\rangle_U,$ namely the transformation $\beta$ is constructed. By repeating the procedure $\mathcal{S}$, $\beta$ can be obtained exponentially fast.

$\beta^{-1}$ can be constructed in a similar way by repeated use of the following procedure $\mathcal{T}$:

\begin{equation}
\xymatrix{
|H\rangle_B|i\rangle_U \ar[r]^-{R} & \frac{1}{\sqrt{2}}(|H\rangle_A|-i\rangle_W + |H\rangle_B|-i\rangle_V) \ar[dl]_-{\mathcal{M}_1} & \\
output \ar[r]_-{Pr = \frac{1}{2}} \ar[d]^-{Pr = \frac{1}{2}} & |H\rangle_A|-i\rangle_W \ar[r]^-{Id \otimes p^2q^2p^2} & |H\rangle_A|i\rangle_W \\
|H\rangle_B|-i\rangle_V  \ar[r]^-{Id \otimes p^2q^2p^2} & |H\rangle_B|i\rangle_V \ar[r]^-{Id \otimes \gamma^{-1}} &  |H\rangle_B|i\rangle_U
}
\end{equation}

\end{proof}
\end{lem}

By going back and forth between $W$ and $U$ via $\beta$ and $\beta^{-1}$, any operation in the space $U$ can be performed in $W$ accordingly. In particular, the Controlled-$Z$ gate and Measurement \ref{Meas 3} can be constructed in $W$.

Collecting the results in this subsection, we finish the proof of Theorem \ref{W universal}.

\appendix

\section{Solutions of the $D(\Sym_3)$ fusion rules} \label{FRmatrices}

Given a set of fusion rules, it is highly non-trivial to solve for all $6j$ symbols even with software packages.  Though $D(\Sym_3)$ is a large anyon system, recent progress makes it possible to solve for all modular categories with the same fusion rules.  In the following, we list the complete data only for the $D(\Sym_3)$. See Section \ref{FusionTreeBasis} for an explanation of the notations that we use for $F$-matrices and $R$-matrices.

The subcategory spanned by $\{ A, B, G\}$ is a near-group category of type $(\mathbb{Z}_2, 1)$ and analyzed completely in \cite{siehler2003} (The objects $A, B, G$ here are the $\epsilon, g, m$ in \cite{siehler2003}, respectively.)

 The only monoidal structure which allows braiding is the following one \cite{siehler2003}.  When we list $6j$ symbols, all the admissible ones that are equal to $1$ are omitted.

 $F^{BGG}_{GGG}=F^{GBG}_{GGG}=F^{GGB}_{GGG}=F^{GGG}_{BGG}=-1$,

 $F^{GGG}_{G}=\begin{pmatrix}1/2 & 1/2&1/\sqrt{2} \\1/2&1/2 &-1/\sqrt{2} \\1/\sqrt{2}&-1/\sqrt{2} & 0 \end{pmatrix}$

 Note that we normalize the trivalent basis to obtain unitary $F$ matrices, while the original $F$ matrices in \cite{siehler2003} are not unitary.

There are two braiding structures on the subcategory depending on a choice of $\omega \in \{e^{2\pi i/3}, e^{4\pi i/3} \}$.

 $R^{AA}_{A}=R^{AB}_{B}=R^{BA}_{A}=R^{AG}_{G}=R^{GA}_{G}=1,$

 $R^{BG}_{G}=R^{GB}_{G}=-1,$

 $R^{GG}_{A}=\omega^2,$

 $R^{GG}_{B}=-\omega^2,$

 $R^{GG}_{G}=\omega .$

 Since these two structures are complex conjugate to each other, we assume  that $\omega = e^{2\pi i/3}$ from now on. The subcategory is balanced for all choices of $\omega \in \{1, e^{2\pi i/3}, e^{4\pi i/3} \}$.  Siehler chose $\omega=1$, which does not extend to the whole category.

 There are three monoidal structures that extend to all other simple objects. We will focus on the structure that we used in this paper. For simplicity, let $\mathcal{G}=\{A, B\}$, $\mathcal{C}_1=\{G\}$, $\mathcal{C}_2=\{D, E\}$, and $\mathcal{C}_3=\{C, F, H\}$. In the following we list associativity matrices according to types upon three upper objects in $F^{abc}_d$. For example, $\mathcal{G}\mathcal{G}\mathcal{C}_2-$ type contains all associativity matrices with two objects from $\mathcal{G}$ and one from $\mathcal{C}_2$. Using this notation, all $\mathcal{G}\mathcal{G}\mathcal{G}-$, $\mathcal{G}\mathcal{G}\mathcal{C}_1-$, $\mathcal{G}\mathcal{C}_1\mathcal{C}_1-$, and $\mathcal{C}_1\mathcal{C}_1\mathcal{C}_1-$ types are given above.

\subsection{The rest of the $6j$ symbols} Beside the $6j$-symbols above, the rest are:

 $\mathcal{G}\mathcal{G}\mathcal{C}_2-$ type:

\begin{itemize}

 \item $F^{BDB}_{D}=F^{BEB}_{E}=-1$
\end{itemize}

 $\mathcal{G}\mathcal{C}_1\mathcal{C}_2-$ type:

 \begin{itemize}

 \item $ -1\:\: \text{for}\:\:F^{BDG}_{E}, F^{BEG}_{D}, F^{GDB}_{E}, F^{GEB}_{D}$

 \end{itemize}

 $\mathcal{G}\mathcal{C}_1\mathcal{C}_3-$ type:

\begin{itemize}

\item $-1$ for \\$
F^{BGC}_{H},F^{BGH}_{F},F^{BGH}_{C},F^{BFG}_{C},F^{BCG}_{F},F^{BCG}_{H},F^{BHG}_{C},F^{GBF}_{H},F^{GBC}_{H},\\
F^{GBH}_{C},F^{GFB}_{C}, F^{GCB}_{F},F^{GCB}_{H},F^{GHB}_{C},F^{FBG}_{H},F^{CBG}_{H},F^{CGB}_{H},F^{HBG}_{C},\\
F^{HGB}_{F},F^{HGB}_{C} $

\end{itemize}

$\mathcal{G}\mathcal{C}_2\mathcal{C}_2-$ type:

\begin{itemize}

\item $                -1$ for \\$
F^{BDD}_{F},F^{BED}_{F},F^{DBD}_{B},F^{DBE}_{G},F^{DBE}_{C},F^{DBE}_{H},F^{DDB}_{F},F^{DEB}_{F},F^{EBD}_{G},\\
F^{EBD}_{C},F^{EBD}_{H},F^{EBE}_{B} $

\end{itemize}

$\mathcal{G}\mathcal{C}_2\mathcal{C}_3-$ type:

\begin{itemize}

\item $                -1$ for \\$
F^{BDC}_{E},F^{BDH}_{E},F^{BEC}_{D},F^{BEH}_{D},F^{BFD}_{D},F^{BFD}_{E},F^{DBF}_{D},F^{DFB}_{D},F^{DFB}_{E},\\
F^{EBF}_{D},F^{FBD}_{D},F^{FBE}_{D},F^{CDB}_{E},F^{CEB}_{D},F^{HDB}_{E},F^{HEB}_{D}  $

\end{itemize}

$\mathcal{G}\mathcal{C}_3\mathcal{C}_3-$ type:

\begin{itemize}

\item $                -1$ for \\$
F^{BFF}_{F},F^{BFC}_{H},F^{BFH}_{G},F^{BCC}_{C},F^{BCH}_{G},F^{BHC}_{G},F^{BHC}_{F},F^{BHH}_{H},F^{FBF}_{F},\\
F^{FBC}_{G},F^{FBH}_{C},F^{FFB}_{F},F^{CBF}_{G},F^{CBC}_{C},F^{CBH}_{G},F^{CFB}_{H},F^{CCB}_{C},F^{CHB}_{G},\\
F^{CHB}_{F},F^{HBF}_{C},F^{HBC}_{G},F^{HBH}_{H},F^{HFB}_{G},F^{HCB}_{G},F^{HHB}_{H} $

\end{itemize}

$\mathcal{C}_1\mathcal{C}_1\mathcal{C}_2-$ type:

\begin{itemize}

\item $ \frac{1}{\sqrt{2}}\begin{pmatrix} 1&1 \\1 & -1 \end{pmatrix}\:\: \text{for}\:\:
F^{G,G,D}_{D},F^{GGE}_{E}, F^{DGG}_{D},F^{EGG}_{E}$

\item $                \frac{1}{\sqrt{2}}\begin{pmatrix} 1& -1\\1 & 1 \end{pmatrix}\:\:\text{for}\:\:
 F^{G,G,D}_{E},F^{GGE}_{D} $

\item $                \frac{1}{\sqrt{2}}\begin{pmatrix} 1& 1\\-1 & 1 \end{pmatrix}\:\:\text{for}\:\:
 F^{DGG}_{E},F^{EGG}_{D} $

\item $  \frac{1}{2}\begin{pmatrix} -1& -\sqrt{3}\\-\sqrt{3} & 1 \end{pmatrix}\:\:\text{for}\:\:
 F^{GDG}_{D}$

\item $  \frac{1}{2}\begin{pmatrix} -\sqrt{3}& 1\\1 &  \sqrt{3} \end{pmatrix}\:\:\text{for}\:\:
 F^{GDG}_{E},F^{GEG}_{D} $

\item $  \frac{1}{2}\begin{pmatrix} 1&  \sqrt{3}\\ \sqrt{3} & -1 \end{pmatrix}\:\:\text{for}\:\:
 F^{GEG}_{E}$

 \end{itemize}

\vspace{10mm}
$\mathcal{C}_1\mathcal{C}_1\mathcal{C}_3-$ type:

\vspace{5mm}
\begin{itemize}
\item $ \frac{1}{\sqrt{2}}\begin{pmatrix} 1&1 \\1 & -1 \end{pmatrix} \:\: \text{for}\:\:
F^{GGF}_{F},F^{GGC}_{C},F^{GGH}_{H},F^{FGG}_{F},F^{CGG}_{C},F^{HGG}_{H}$

\item $               \begin{pmatrix} 0&1 \\1 & 0 \end{pmatrix} \:\:\text{for}\:\:
  F^{GFG}_{F},F^{GCG}_{C},F^{GHG}_{H}$

 \end{itemize}

$\mathcal{C}_1\mathcal{C}_2\mathcal{C}_2-$ type:

\begin{itemize}

\item $ -1 \:\: \text{for}\:\:
F^{DGE}_{B},F^{EGD}_{B}$

\item $           \frac{1}{\sqrt{2}}\begin{pmatrix} 1&1 \\1 & -1 \end{pmatrix}    $ for \\$
  F^{GDD}_{G},F^{GDD}_{C},F^{GDD}_{H},F^{GEE}_{G},F^{GEE}_{C},F^{GEE}_{H},F^{DDG}_{G},F^{DDG}_{C},F^{DDG}_{H},\\
  F^{EEG}_{G},F^{EEG}_{C},F^{EEG}_{H}$

\item $           \frac{1}{\sqrt{2}}\begin{pmatrix} 1&-1 \\1 & 1 \end{pmatrix}      $ for \\$
 F^{GDD}_{F},F^{GED}_{F},F^{DEG}_{G},F^{DEG}_{H},F^{EDG}_{G},F^{EDG}_{F},F^{EDG}_{C},F^{EEG}_{F}$

\item $           \frac{1}{\sqrt{2}}\begin{pmatrix} 1&1 \\-1 & 1 \end{pmatrix}      $ for \\$
 F^{GDE}_{G},F^{GDE}_{F},F^{GDE}_{C},F^{GED}_{G},F^{GED}_{H},F^{GEE}_{F},F^{DDG}_{F},F^{DEG}_{F}$

\item $           \frac{1}{\sqrt{2}}\begin{pmatrix} -1&-1 \\-1 & 1 \end{pmatrix}       \:\:\text{for}\:\:
 F^{GDE}_{H},F^{GED}_{C},F^{DEG}_{C},F^{EDG}_{H}$

\item $           \frac{1}{2}\begin{pmatrix} -1&-\sqrt{3} \\-\sqrt{3} & 1 \end{pmatrix}       \:\:\text{for}\:\:
 F^{DGD}_{G},F^{DGD}_{F}$

\item $           \frac{1}{2}\begin{pmatrix} -1& \sqrt{3} \\ \sqrt{3} & 1 \end{pmatrix}       \:\:\text{for}\:\:
 F^{DGD}_{C}$

\item $           \begin{pmatrix} 1&0 \\0 & -1 \end{pmatrix}       \:\:\text{for}\:\:
 F^{DGD}_{H}$

\item $           \frac{1}{2}\begin{pmatrix} -\sqrt{3}&1 \\1 & \sqrt{3} \end{pmatrix}       \:\:\text{for}\:\:
 F^{DGE}_{G},F^{EGD}_{G}$

\item $           \frac{1}{2}\begin{pmatrix} -\sqrt{3}&-1 \\1 & -\sqrt{3} \end{pmatrix}       \:\:\text{for}\:\:
 F^{DGE}_{F}$

\item $           \frac{1}{2}\begin{pmatrix} \sqrt{3}&1 \\1 & -\sqrt{3} \end{pmatrix}       \:\:\text{for}\:\:
 F^{DGE}_{C},F^{EGD}_{C}$

\item $           \begin{pmatrix} 0&-1 \\-1 & 0 \end{pmatrix}       \:\:\text{for}\:\:
 F^{DGE}_{H},F^{EGD}_{H}$

\item $           \frac{1}{2}\begin{pmatrix} -\sqrt{3}&1 \\-1 & -\sqrt{3} \end{pmatrix}       \:\:\text{for}\:\:
 F^{EGD}_{F}$

\item $           \frac{1}{2}\begin{pmatrix} 1&\sqrt{3} \\ \sqrt{3} & -1 \end{pmatrix}       \:\:\text{for}\:\:
 F^{EGE}_{G}$

\item $           \frac{1}{2}\begin{pmatrix} 1&-\sqrt{3} \\-\sqrt{3} & -1 \end{pmatrix}       \:\:\text{for}\:\:
 F^{EGE}_{F},F^{EGE}_{C}$

 \item $          \begin{pmatrix} -1&0 \\0 & 1 \end{pmatrix}       \:\:\text{for}\:\:
 F^{EGE}_{H}$
 \end{itemize}

$\mathcal{C}_1\mathcal{C}_2\mathcal{C}_3-$ type:

\begin{itemize}

\item $ \frac{1}{2}\begin{pmatrix} -1&-\sqrt{3} \\-\sqrt{3} & 1 \end{pmatrix}  \:\: \text{for}\:\:
F^{GDF}_{D},F^{FDG}_{D}$

\item $   \frac{1}{2}\begin{pmatrix} -\sqrt{3}&-1 \\1 & -\sqrt{3} \end{pmatrix}             \:\:\text{for}\:\:
 F^{GDF}_{E},F^{FEG}_{D} $

\item $   \frac{1}{2}\begin{pmatrix} -1& \sqrt{3} \\ \sqrt{3} & 1   \end{pmatrix}          \:\:\text{for}\:\:
 F^{GDC}_{D},F^{CDG}_{D}$

 \item $  \frac{1}{2}\begin{pmatrix}  \sqrt{3}&1 \\1 & -\sqrt{3} \end{pmatrix}             \:\:\text{for}\:\:
 F^{GDC}_{E},F^{GEC}_{D},F^{CDG}_{E},F^{CEG}_{D}$

 \item $          \begin{pmatrix} 1&0 \\0 & -1 \end{pmatrix}       \:\:\text{for}\:\:
 F^{GDH}_{D},F^{HDG}_{D}$

 \item $          \begin{pmatrix} 0&-1 \\-1 & 0 \end{pmatrix}       \:\:\text{for}\:\:
 F^{GDH}_{E},F^{GEH}_{D},F^{HDG}_{E},F^{HEG}_{D}$

 \item $  \frac{1}{2}\begin{pmatrix} -\sqrt{3}&1 \\-1 & -\sqrt{3} \end{pmatrix}             \:\:\text{for}\:\:
 F^{GEF}_{D},F^{FDG}_{E}$

\item $   \frac{1}{2}\begin{pmatrix} 1&-\sqrt{3} \\-\sqrt{3} & -1   \end{pmatrix}          \:\:\text{for}\:\:
 F^{GEF}_{E},F^{GEC}_{E},F^{FEG}_{E},F^{CEG}_{E}$

 \item $          \begin{pmatrix} -1&0 \\0 & 1 \end{pmatrix}       \:\:\text{for}\:\:
 F^{GEH}_{E},F^{HEG}_{E}$

\item $           \frac{1}{\sqrt{2}}\begin{pmatrix} 1&1 \\-1 & 1 \end{pmatrix}       $ for \\$
 F^{GFD}_{D},F^{GFD}_{E},F^{DGF}_{E},F^{DGC}_{E},F^{DHG}_{E},F^{EGF}_{E},F^{EGH}_{D},F^{EFG}_{D},F^{EFG}_{E},\\
 F^{ECG}_{D},F^{FGD}_{D},F^{FGE}_{D} $

 \item $          \frac{1}{\sqrt{2}}\begin{pmatrix} 1&-1 \\1 & 1 \end{pmatrix}       $ for \\$
 F^{GFE}_{D},F^{GFE}_{E},F^{GCE}_{D},F^{GHD}_{E},F^{DGF}_{D},F^{DFG}_{D},F^{DFG}_{E},F^{EGF}_{D},F^{FGD}_{E},\\
 F^{FGE}_{E},F^{CGD}_{E},F^{HGE}_{D}$

 \item $          \frac{1}{\sqrt{2}}\begin{pmatrix} 1&1 \\1 & -1 \end{pmatrix}       $ for \\$
 F^{GCD}_{D},F^{GCE}_{E},F^{GHD}_{D},F^{GHE}_{E},F^{DGC}_{D},F^{DGH}_{D},F^{DCG}_{D},F^{DHG}_{D},F^{EGC}_{E},\\
 F^{EGH}_{E},F^{ECG}_{E},F^{EHG}_{E},F^{CGD}_{D},F^{CGE}_{E},F^{HGD}_{D},F^{HGE}_{E}$

 \item $          \frac{1}{\sqrt{2}}\begin{pmatrix} -1&-1 \\-1 & 1 \end{pmatrix}      $ for \\$
 F^{GCD}_{E},F^{GHE}_{D},F^{DGH}_{E},F^{DCG}_{E},F^{EGC}_{D},F^{EHG}_{D},F^{CGE}_{D},F^{HGD}_{E}  $ \end{itemize}

$\mathcal{C}_1\mathcal{C}_3\mathcal{C}_3-$ type:

\begin{itemize}

\item $ \frac{1}{\sqrt{2}}\begin{pmatrix} 1&1 \\1 & -1 \end{pmatrix}       \:\:\text{for}\:\:
F^{GFF}_{G},F^{GCC}_{G},F^{GHH}_{G},F^{FFG}_{G},F^{CCG}_{G},F^{HHG}_{G}$

\item $ -1 $ for \\$
F^{GCH}_{B},F^{GHF}_{B},F^{GHC}_{B},F^{FGC}_{B},F^{FHG}_{B},F^{CGF}_{B},F^{CGH}_{B},F^{CHG}_{B},F^{HGC}_{B},\\
F^{HCG}_{B}$

\item $ \begin{pmatrix} 0&1 \\1 & 0 \end{pmatrix}       \:\:\text{for}\:\:
F^{FGF}_{G},F^{CGC}_{G},F^{HGH}_{G}$

\end{itemize}

$\mathcal{C}_2\mathcal{C}_2\mathcal{C}_2-$ type:

\begin{itemize}

\item $ \frac{1}{3}\begin{pmatrix}
1       &\sqrt{2}&\sqrt{2}&\sqrt{2}&\sqrt{2}\\
\sqrt{2}&      -1&-1      &-1      &2\\
\sqrt{2}&      -1&2       &-1      &-1\\
\sqrt{2}&      -1&-1      &2       &-1\\
\sqrt{2}&       2&-1      &-1      &-1 \end{pmatrix}       \:\:\text{for}\:\:
F^{DDD}_{D},F^{EEE}_{E}$

\item $ \frac{1}{\sqrt{3}}\begin{pmatrix}
-1&-1&1&0\\
-1&0&-1&-1\\
1&-1&0&-1\\
0&-1&-1&1
\end{pmatrix}       \:\:\text{for}\:\:
F^{DDD}_{E},F^{DDE}_{D},F^{DED}_{D},F^{EDD}_{D}$

\item $ \frac{1}{3}\begin{pmatrix}
1 &-\sqrt{2}&\sqrt{2}&-\sqrt{2}&-\sqrt{2}\\
\sqrt{2}&1&-1      &1&-2\\
\sqrt{2}&1&2       &1&1\\
\sqrt{2}&1&-1      &-2&1\\
\sqrt{2}&-2&-1       &1&1 \end{pmatrix}       \:\:\text{for}\:\:
F^{DDE}_{E},F^{EED}_{D}$

\item $ \frac{1}{3}\begin{pmatrix}
-1      &\sqrt{2}&\sqrt{2}&\sqrt{2}&\sqrt{2}\\
\sqrt{2}&        1&1      &1        &-2\\
\sqrt{2}&        1&-2     &1        &1\\
\sqrt{2}&        1&1      &-2       &1\\
\sqrt{2}&       -2&1      &1        &1 \end{pmatrix}       \:\:\text{for}\:\:
F^{DED}_{E},F^{EDE}_{D}$

\item $ \frac{1}{3}\begin{pmatrix}
1       &\sqrt{2}&\sqrt{2}&\sqrt{2}&\sqrt{2}\\
-\sqrt{2}&      1&1      &1      &-2\\
\sqrt{2}&      -1&2       &-1      &-1\\
-\sqrt{2}&      1&1      &-2       &1\\
-\sqrt{2}&       -2&1      &1      &1 \end{pmatrix}       \:\:\text{for}\:\:
F^{DEE}_{D},F^{EDD}_{E}$

\item $ \frac{1}{\sqrt{3}}\begin{pmatrix}
1&-1&-1&0\\
-1&0&-1&-1\\
-1&-1&0&1\\
0&-1&1&-1
\end{pmatrix}       \:\:\text{for}\:\:
F^{DEE}_{E},F^{EDE}_{E},F^{EED}_{E},F^{EEE}_{D}$
\end{itemize}

$\mathcal{C}_2\mathcal{C}_2\mathcal{C}_3-$ type:

\begin{itemize}

\item $-1 $ for \\$
F^{DDF}_{B},F^{DCE}_{B},F^{DHE}_{B},F^{EDF}_{B},F^{ECD}_{B},F^{EHD}_{B},F^{FDD}_{B},F^{FDE}_{B}$

\item $\frac{1}{\sqrt{2}}\begin{pmatrix}
 1&1 \\
 -1& 1  \end{pmatrix}  $ for \\$
F^{DDF}_{G},F^{DDF}_{H},F^{DDC}_{H},F^{DDH}_{F},F^{DDH}_{C},F^{DEH}_{F},F^{EDF}_{G},F^{EDF}_{H},F^{EEC}_{H},\\
F^{EEH}_{C},F^{FED}_{G},F^{FED}_{H},F^{FEE}_{G},F^{FEE}_{F},F^{FEE}_{H},F^{CDE}_{C},F^{CED}_{G},F^{CED}_{C},\\
F^{HDE}_{G},F^{HDE}_{F},F^{HDE}_{H},F^{HED}_{H},F^{HEE}_{F}$

\item $\frac{1}{\sqrt{2}}\begin{pmatrix}
 1&1 \\
 1& -1  \end{pmatrix}  $ for \\$
F^{DDF}_{F},F^{DDF}_{C},F^{DDC}_{G},F^{DDC}_{F},F^{DDC}_{C},F^{DDH}_{G},F^{DDH}_{H},F^{DEC}_{F},F^{EDF}_{C},\\
F^{EEC}_{G},F^{EEC}_{C},F^{EEH}_{G},F^{EEH}_{H},F^{FDD}_{F},F^{FDD}_{C},F^{FDE}_{C},F^{CDD}_{G},F^{CDD}_{F},\\
F^{CDD}_{C},F^{CED}_{F},F^{CEE}_{G},F^{CEE}_{C},F^{HDD}_{G},F^{HDD}_{H},F^{HEE}_{G},F^{HEE}_{H}$

\item $\frac{1}{\sqrt{2}}\begin{pmatrix}
 1&-1 \\
 1& 1  \end{pmatrix}  $ for \\$
F^{DEF}_{G},F^{DEF}_{H},F^{DEC}_{G},F^{DEC}_{C},F^{DEH}_{H},F^{EDC}_{C},F^{EDH}_{G},F^{EDH}_{F},F^{EDH}_{H},\\
F^{EEF}_{G},F^{EEF}_{F},F^{EEF}_{H},F^{EEH}_{F},F^{FDD}_{G},F^{FDD}_{H},F^{FDE}_{G},F^{FDE}_{H},F^{CDD}_{H},\\
F^{CEE}_{H},F^{HDD}_{F},F^{HDD}_{C},F^{HED}_{F},F^{HEE}_{C}$

\item $\frac{1}{\sqrt{2}}\begin{pmatrix}
 -1&-1 \\
 1& -1  \end{pmatrix}  $ for \\$
F^{DEF}_{F},F^{DEC}_{H},F^{EDF}_{F},F^{EDH}_{C}$

\item $\frac{1}{\sqrt{2}}\begin{pmatrix}
 -1&1 \\
 1& 1  \end{pmatrix}  $ for \\$
F^{DEF}_{C},F^{DEH}_{C},F^{EDC}_{F},F^{EDC}_{H},F^{EEF}_{C},F^{EEC}_{F},F^{FED}_{C},F^{FEE}_{C},F^{CDE}_{F},\\
F^{CDE}_{H},F^{CEE}_{F},F^{HED}_{C}$

\item $\frac{1}{\sqrt{2}}\begin{pmatrix}
 -1&-1 \\
 -1& 1  \end{pmatrix}  $ for \\$
F^{DEH}_{G},F^{EDC}_{G},F^{CDE}_{G},F^{HED}_{G}$

\item $\frac{1}{2}\begin{pmatrix}
 -1&             -\sqrt{3} \\
 -\sqrt{3}& 1  \end{pmatrix}  $ for \\$
F^{DFD}_{G},F^{DFD}_{C},F^{DFD}_{H},F^{DCD}_{F},F^{DCD}_{H},F^{DHD}_{F},F^{DHD}_{C}$

\item $\begin{pmatrix}
 1& 0 \\
 0& -1  \end{pmatrix}  \:\:\text{for}\:\:
F^{DFD}_{F},F^{DCD}_{C},F^{DHD}_{G}$

\item $\frac{1}{2}\begin{pmatrix}
 -\sqrt{3}&1 \\
 -1& -\sqrt{3}  \end{pmatrix}  \:\:\text{for}\:\:
F^{DFE}_{G},F^{DFE}_{C},F^{DFE}_{H},F^{ECD}_{F},F^{EHD}_{F}$

\item $\begin{pmatrix}
 0& 1 \\
 1& 0  \end{pmatrix}  \:\:\text{for}\:\:
F^{DFE}_{F},F^{EFD}_{F}$

\item $\frac{1}{2}\begin{pmatrix}
 -1&             \sqrt{3} \\
 \sqrt{3}& 1  \end{pmatrix}  \:\:\text{for}\:\:
F^{DCD}_{G},F^{DHD}_{H}$

\item $\frac{1}{2}\begin{pmatrix}
 \sqrt{3}&  1 \\
 1&              -\sqrt{3}  \end{pmatrix}  \:\:\text{for}\:\:
F^{DCE}_{G},F^{DHE}_{H},F^{ECD}_{G},F^{EHD}_{H}$

\item $\frac{1}{2}\begin{pmatrix}
 -\sqrt{3}&-1 \\
 1& -\sqrt{3}  \end{pmatrix}  \:\:\text{for}\:\:
F^{DCE}_{F},F^{DHE}_{F},F^{EFD}_{G},F^{EFD}_{C},F^{EFD}_{H}$

\item $\begin{pmatrix}
 0& -1 \\
 -1& 0  \end{pmatrix}  \:\:\text{for}\:\:
F^{DCE}_{C},F^{DHE}_{G},F^{ECD}_{C},F^{EHD}_{G}$

\item $\frac{1}{2}\begin{pmatrix}
 -\sqrt{3}&1 \\
 1& \sqrt{3}  \end{pmatrix}  \:\:\text{for}\:\:
F^{DCE}_{H},F^{DHE}_{C},F^{ECD}_{H},F^{EHD}_{C}$

\item $\frac{1}{2}\begin{pmatrix}
 1&-\sqrt{3} \\
 -\sqrt{3}&-1  \end{pmatrix}  $ for \\$
F^{EFE}_{G},F^{EFE}_{C},F^{EFE}_{H},F^{ECE}_{G},F^{ECE}_{F},F^{EHE}_{F},F^{EHE}_{H}$

\item $\begin{pmatrix}
 -1& 0 \\
 0& 1  \end{pmatrix}  \:\:\text{for}\:\:
F^{EFE}_{F},F^{ECE}_{C},F^{EHE}_{G}$

\item $\frac{1}{2}\begin{pmatrix}
 1& \sqrt{3} \\
  \sqrt{3}&-1  \end{pmatrix}  \:\:\text{for}\:\:
F^{ECE}_{H},F^{EHE}_{C}$

\item $\frac{1}{\sqrt{2}}\begin{pmatrix}
 -1& 1 \\
 -1& -1  \end{pmatrix}  \:\:\text{for}\:\:
F^{FDE}_{F},F^{FED}_{F},F^{CED}_{H},F^{HDE}_{C} $

\end{itemize}

$\mathcal{C}_2\mathcal{C}_3\mathcal{C}_3-$ type:

\begin{itemize}

\item $\frac{1}{\sqrt{2}}\begin{pmatrix}
 1&1 \\
 1& -1  \end{pmatrix} $ for \\$
 F^{DFF}_{D},F^{DFC}_{D},F^{DFC}_{E},F^{DCF}_{D},F^{DCC}_{D},F^{DHH}_{D},F^{ECF}_{D},F^{ECC}_{E},F^{EHH}_{E},\\
 F^{FFD}_{D},F^{FCD}_{D},F^{FCE}_{D},F^{CFD}_{D},F^{CFD}_{E},F^{CCD}_{D},F^{CCE}_{E},F^{HHD}_{D},F^{HHE}_{E}$

\item $\frac{1}{\sqrt{2}}\begin{pmatrix}
 -1&1 \\
 -1& -1  \end{pmatrix} $ for \\$
 F^{DFF}_{E},F^{DHC}_{E},F^{EFF}_{D},F^{ECH}_{D}$

\item $\frac{1}{\sqrt{2}}\begin{pmatrix}
 1&-1 \\
 1& 1  \end{pmatrix} $ for \\$
 F^{DFH}_{D},F^{DFH}_{E},F^{DCH}_{D},F^{DHF}_{D},F^{DHC}_{D},F^{ECH}_{E},F^{EHF}_{D},F^{EHC}_{E},F^{FFE}_{E},\\
 F^{FHD}_{E},F^{FHE}_{E},F^{CCD}_{E},F^{CCE}_{D},F^{HFE}_{D},F^{HFE}_{E},F^{HHD}_{E},F^{HHE}_{D}$

\item $\frac{1}{\sqrt{2}}\begin{pmatrix}
 -1&1 \\
 1& 1  \end{pmatrix} $ for \\$
 F^{DCF}_{E},F^{DCH}_{E},F^{EFC}_{D},F^{EFC}_{E},F^{ECF}_{E},F^{EHC}_{D},F^{FCD}_{E},F^{FCE}_{E},F^{CFE}_{D},\\
 F^{CFE}_{E},F^{CHE}_{D},F^{HCD}_{E}$

\item $\frac{1}{\sqrt{2}}\begin{pmatrix}
 1&1 \\
 -1& 1  \end{pmatrix} $ for \\$
 F^{DCC}_{E},F^{DHF}_{E},F^{DHH}_{E},F^{EFF}_{E},F^{EFH}_{D},F^{EFH}_{E},F^{ECC}_{D},F^{EHF}_{E},F^{EHH}_{D},\\
 F^{FHD}_{D},F^{FHE}_{D},F^{CHD}_{D},F^{CHE}_{E},F^{HFD}_{D},F^{HFD}_{E},F^{HCD}_{D},F^{HCE}_{E}$

\item $ \begin{pmatrix}
 1&0 \\
 0& -1  \end{pmatrix} \:\:\text{for}\:\:
 F^{FDF}_{D},F^{CDC}_{D}$

\item $ \begin{pmatrix}
 0&1 \\
 1&0  \end{pmatrix} \:\:\text{for}\:\:
 F^{FDF}_{E},F^{FEF}_{D}$

\item $ \frac{1}{2}\begin{pmatrix}
 -1&-\sqrt{3} \\
 -\sqrt{3}&1  \end{pmatrix}  $ for \\$
F^{FDC}_{D},F^{FDH}_{D},F^{CDF}_{D},F^{CDH}_{D},F^{HDF}_{D},F^{HDC}_{D}$

\item $ \frac{1}{2}\begin{pmatrix}
 -\sqrt{3}&1 \\
 -1&-\sqrt{3}  \end{pmatrix}  \:\:\text{for}\:\:
F^{FDC}_{E},F^{FDH}_{E},F^{CEF}_{D},F^{HEF}_{D}$

\item $ \begin{pmatrix}
 -1&0 \\
 0&1  \end{pmatrix} \:\:\text{for}\:\:
 F^{FEF}_{E},F^{CEC}_{E}$

\item $\frac{1}{2}\begin{pmatrix}
 -\sqrt{3}&-1 \\
 1&-\sqrt{3}  \end{pmatrix}  \:\:\text{for}\:\:
F^{FEC}_{D},F^{FEH}_{D},F^{CDF}_{E},F^{HDF}_{E}$

\item $\frac{1}{2}\begin{pmatrix}
 1&-\sqrt{3} \\
 -\sqrt{3}&-1  \end{pmatrix}  \:\:\text{for}\:\:
F^{FEC}_{E},F^{FEH}_{E},F^{CEF}_{E},F^{HEF}_{E},F^{HEH}_{E}$

\item $\frac{1}{\sqrt{2}}\begin{pmatrix}
 -1&-1 \\
 1& -1  \end{pmatrix} \:\:\text{for}\:\:
 F^{FFD}_{E},F^{FFE}_{D},F^{CHD}_{E},F^{HCE}_{D}$

\item $\begin{pmatrix}
 0&-1 \\
 -1&0  \end{pmatrix} \:\:\text{for}\:\:
 F^{CDC}_{E},F^{CEC}_{D}$

 \item $\frac{1}{2}\begin{pmatrix}
 -\sqrt{3}&1 \\
 1&\sqrt{3}  \end{pmatrix}  \:\:\text{for}\:\:
F^{CDH}_{E},F^{CEH}_{D},F^{HDC}_{E},F^{HEC}_{D}$

\item $\frac{1}{2}\begin{pmatrix}
 1&\sqrt{3} \\
  \sqrt{3}&-1  \end{pmatrix}  \:\:\text{for}\:\:
F^{CEH}_{E},F^{HEC}_{E}$

\item $\frac{1}{2}\begin{pmatrix}
 -1& \sqrt{3} \\
  \sqrt{3}&1  \end{pmatrix}  \:\:\text{for}\:\:
F^{HDH}_{D}$

\item $\frac{1}{2}\begin{pmatrix}
  \sqrt{3}&1 \\
 1&- \sqrt{3}  \end{pmatrix}  \:\:\text{for}\:\:
F^{HDH}_{E},F^{HEH}_{D}$

\end{itemize}

$\mathcal{C}_3\mathcal{C}_3\mathcal{C}_3-$ type:

\begin{itemize}

\item $-1 \:\:\text{for}\:\:
F^{FFF}_{B},F^{FCH}_{B},F^{CCC}_{B},F^{HCF}_{B},F^{HHH}_{B}$

\item  $\frac{1}{2}\begin{pmatrix}
1&       1        &\sqrt{2}\\
1&       1        &-\sqrt{2}\\
\sqrt{2}&-\sqrt{2}& 0  \end{pmatrix}  \:\:\text{for}\:\:
F^{FFF}_{F},F^{CCC}_{C},F^{HHH}_{H}$

\item  $\frac{1}{\sqrt{2}}\begin{pmatrix}
 1&-1 \\
 1& 1  \end{pmatrix} \:\:\text{for}\:\:
F^{FFC}_{C},F^{CCF}_{F},F^{CCH}_{H},F^{HHC}_{C}$

\item  $\frac{1}{\sqrt{2}}\begin{pmatrix}
 1&1 \\
 1&-1  \end{pmatrix} \:\:\text{for}\:\:
F^{FFH}_{H},F^{FHH}_{F},F^{HFF}_{H},F^{HHF}_{F}$

\item  $\begin{pmatrix}
 0& 1 \\
 1& 0  \end{pmatrix} \:\:\text{for}\:\:
F^{FCF}_{C},F^{FHF}_{H},F^{CFC}_{F},F^{CHC}_{H},F^{HFH}_{F},F^{HCH}_{C}$

\item  $\frac{1}{\sqrt{2}}\begin{pmatrix}
 1&1 \\
 -1&1  \end{pmatrix} \:\:\text{for}\:\:
F^{FCC}_{F},F^{CFF}_{C},F^{CHH}_{C},F^{HCC}_{H}$

\end{itemize}

\subsection{The rest of $R$-symbols}Beside the $R$-symbols at the beginning of the section, the rest are:

\begin{itemize}

\item $1$ for \\
$R^{BB}_{A},R^{GH}_{F},R^{GH}_{C},R^{DF}_{D},R^{DC}_{D},R^{EF}_{E},R^{EC}_{E},R^{FD}_{D},R^{FE}_{E},R^{FF}_{A},R^{FF}_{F},\\
R^{CD}_{D},R^{CE}_{E},R^{CC}_{A},R^{CC}_{C},R^{HG}_{F},R^{HG}_{C},R^{EE}_{A},R^{EE}_{F},R^{EE}_{C}$

\item$-1$ for\\
$R^{BG}_{G},R^{BF}_{F},R^{BC}_{C},R^{BH}_{H},R^{GB}_{G},R^{FB}_{F},R^{FF}_{B},R^{CB}_{C},R^{CC}_{B},R^{HB}_{H},R^{DD}_{A},\\
R^{DD}_{F},R^{DD}_{C}$

\item$ i$ for $R^{BD}_{E},R^{DB}_{E},R^{EF}_{D},R^{EC}_{D},R^{FE}_{D},R^{CE}_{D},R^{DE}_{B},R^{ED}_{B}$

\item$- i$ for $R^{BE}_{D},R^{DF}_{E},R^{DC}_{E},R^{EB}_{D},R^{FD}_{E},R^{CD}_{E},R^{DE}_{F},R^{DE}_{C},R^{ED}_{F},R^{ED}_{C}$

\item$\omega^2$ for \\
$R^{GG}_{A},R^{GF}_{H},R^{GC}_{H},R^{DH}_{D},R^{EH}_{E},R^{FG}_{H},R^{FC}_{G},R^{FH}_{C},R^{CG}_{H},R^{CF}_{G},R^{CH}_{F},\\
R^{HD}_{D},R^{HE}_{E},R^{HF}_{C},R^{HC}_{F},R^{HH}_{H},R^{EE}_{G}$

\item$-\omega^2$ for $R^{GG}_{B},R^{DD}_{G}$

\item$\omega$ for \\
$R^{GG}_{G},R^{GD}_{D},R^{GE}_{E},R^{GF}_{C},R^{GC}_{F},R^{DG}_{D},R^{EG}_{E},R^{FG}_{C},R^{FC}_{H},R^{FH}_{G},R^{CG}_{F},\\
R^{CF}_{H},R^{CH}_{G},R^{HF}_{G},R^{HC}_{G},R^{HH}_{A},R^{EE}_{H}$

\item$-\omega$ for
$R^{DD}_{H},R^{HH}_{B}$

\item$\omega i$ for
$R^{GE}_{D},R^{EG}_{D}$

\item$-\omega i$ for
$R^{GD}_{E},R^{DG}_{E},R^{DE}_{H},R^{ED}_{H}$

\item$ \omega^2 i$ for
$R^{EH}_{D}, R^{HE}_{D}$

\item$-\omega^2 i$ for
$R^{DE}_{G},R^{ED}_{G},R^{DH}_{E}, R^{HD}_{E}$

\end{itemize}

\section{Representations of braid group $\B_4$} \label{rep of braid}

 In this appendix,  mathematically, we study whether or not the representations of $\B_4$ are irreducible and identify the images of those representation on each irreducible summand.  We will refer to each irreducible summand as a {\it sector}. For our application to anyonic quantum computation, we also determine whether or not there are unitary transformations $($braiding quantum circuits$)$ in the images that are powerful for quantum computation, especially whether or not these circuits lead to a universal gate set.

 We will provide explicitly the braiding matrices for $\sigma_1,\;\sigma_2, \; \sigma_3$, and then compute what is the group generated by them. Without loss of generality, we may multiply the $\sigma_i \, 's$ by a common factor so that they all have determinant $=1$  (Note that all the $\sigma_i \, 's$ are conjugate to each other).  We still denote the new representation by $\rho(m,z).$ We will focus on sectors which are $3$-dimensional. In this case, the images of the representation on such sectors are subgroups of $\textrm{SU}(3)$. As will be seen later, some interesting subgroups of $\textrm{SU}(3)$ will arise as the image.

As in Section \ref{background}, the representations are denoted by $\rho(m,z)$ on the space $V_{z}^{mmmm},$ which corresponds to the following splitting tree:

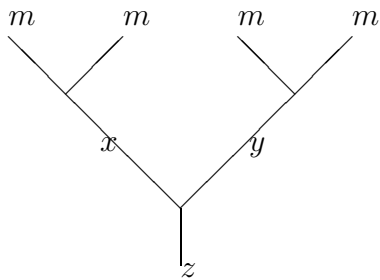
\setlength{\unitlength}{0.030in}
\begin{picture}(160,60)(0,-10)

\put(50,0){\line(0,1){10}}
\put(50,10){\line(1,1){30}}
\put(50,10){\line(-1,1){30}}
\put(70,30){\line(-1,1){10}}
\put(30,30){\line(1,1){10}}

\put(20,42){$m$}
\put(40,42){$m$}
\put(60,42){$m$}
\put(80,42){$m$}
\put(36,20){$x$}
\put(62,20){$y$}
\put(50,-2){$z$}
\end{picture}

Our results are summarized in Tables $3$ and $4$ in Section \ref{background}.  Now we examine each representation explicitly in the following subsections.

\begin{rem} \label{identical particle}
The matrices $\sigma_i \, 's$ depend on the fusion rules of the two anyons $m$, the $6j$-symbols, and the $R$-symbols $R_x^{mm}.$ From this point of view, the anyon $A$ and $B$ are not interesting because their $R$-matrices are trivial.  It follows that their representations are also projectively  trivial. Also, the anyons $C$ and $F$ have identical $R$-matrices. $D$ and $E$ are identical if we multiply the $R$-matrices of $D$ by $-1$. Similarly, $G$ and $H$ are identical if we replace  $\omega=-\frac{1}{2}+\frac{\sqrt{3}i}{2}$ in the $R$-matrices of $G$ by its complex conjugate. Therefore, it suffcies to consider the cases where the anyon $m$ is $C$, $D$, and $G$.
\end{rem}

\subsection{Representations on $V_z^{CCCC}$}

There are three choices for $z$ that make the following splitting tree admissible, namely $A$, $B$ or $C$.

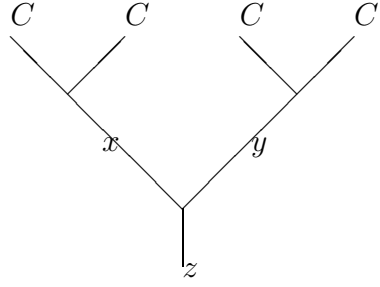
\setlength{\unitlength}{0.030in}
\begin{picture}(160,60)(0,-10)

\put(50,0){\line(0,1){10}}
\put(50,10){\line(1,1){30}}
\put(50,10){\line(-1,1){30}}
\put(70,30){\line(-1,1){10}}
\put(30,30){\line(1,1){10}}

\put(20,42){$C$}
\put(40,42){$C$}
\put(60,42){$C$}
\put(80,42){$C$}
\put(36,20){$x$}
\put(62,20){$y$}
\put(50,-2){$z$}
\end{picture}

\subsubsection{$z=A$}

The basis of $V_A^{CCCC}$ is $\{|AA \rangle, |BB \rangle, |CC \rangle\}$. Under this basis,

$\sigma_1 = \sigma_3 =
\begin{pmatrix}
-1 &  0  &  0 \\
0 &  1 &  0 \\
0 &  0  &  -1 \\
\end{pmatrix}
\qquad
\sigma_2 = -\frac{1}{2}
\begin{pmatrix}
1  &   -1   & \sqrt{2}\\
-1 &   1    & \sqrt{2}\\
\sqrt{2} & \sqrt{2} & 0 \\
\end{pmatrix}
$

This representation splits into two sectors $S_1$ and $S_2$, where $S_1$ is a $1$-dim irrep mapping $\sigma_i$ to $-1$ and $S_2$ is a $2$-dim irrep spanned by $\{|BB\rangle, \frac{\sqrt{3}}{3}|AA\rangle - \frac{\sqrt{6}}{3}|CC\rangle\}.$ The matrices of the $\sigma_i \, 's$ under the basis of $S_2$ are given by:

$\sigma_1 = \sigma_3 = i
\begin{pmatrix}
1 &  0  \\
0 &  -1  \\
\end{pmatrix}
\qquad
\sigma_2 = \frac{i}{2}
\begin{pmatrix}
-1   & -\sqrt{3}\\
-\sqrt{3} & 1 \\
\end{pmatrix}
$

They generate a group which is isomorphic to $\mathbb{Z}_3 \rtimes \mathbb{Z}_4$.

\subsubsection{$z=B$}

The Hilbert space $V_B^{CCCC}$ is also three dimensional with basis $\{|CC \rangle, |AB \rangle, |BA \rangle\}$. The matrices of the $\sigma_i \, 's$ are given by:

$\sigma_1 =
\begin{pmatrix}
-1 &  0  &  0 \\
0 &  -1 &  0 \\
0 &  0  &  1 \\
\end{pmatrix}
\qquad
\sigma_2 = -\frac{1}{2}
\begin{pmatrix}
0  & -\sqrt{2} & -\sqrt{2}\\
-\sqrt{2} & 1  &  -1      \\
-\sqrt{2} & -1  &  1      \\
\end{pmatrix}
$

$
\sigma_3 =
\begin{pmatrix}
-1 &  0  &  0 \\
0 &  1 &  0 \\
0 &  0  &  -1 \\
\end{pmatrix}
$

The representation is irreducible with image isomorphic to the permutation group $\Sym_4$.

\subsubsection{z = C}

The Hilbert space $V_C^{CCCC}$ is five dimensional with basis $\{|CC \rangle, |AC \rangle, |CA \rangle, |BC \rangle, |CB \rangle\}$. And the image of the $\sigma_i \, 's$ are given by :

$\sigma_1 =
\begin{pmatrix}
1 &  0  &  0 & 0 &  0\\
0 &  1 &  0  & 0 &  0\\
0 &  0  & 1  & 0 &  0\\
0 &  0  & 0  & -1 &  0\\
0 &  0  & 0  & 0 &  1\\
\end{pmatrix}
\qquad
\sigma_2 = \frac{1}{2}
\begin{pmatrix}
1 &  0  &  0 & 0 &  0\\
0 &  1 &  1  & 1 &  -1\\
0 &  1  & 1  & -1 &  1\\
0 &  1  & -1  & 1 &  1\\
0 &  -1  & 1  & 1 &  1\\
\end{pmatrix}
$

$
\sigma_3 =
\begin{pmatrix}
1 &  0  &  0 & 0 &  0\\
0 &  1 &  0  & 0 &  0\\
0 &  0  & 1  & 0 &  0\\
0 &  0  & 0  & 1 &  0\\
0 &  0  & 0  & 0 &  -1\\
\end{pmatrix}
\qquad
$

$V_C^{CCCC}$ splits into the direct sum of two trivial irreps and a 3-dim irrep V.  V has a basis $\{\frac{1}{\sqrt{2}}(|AC\rangle - |CA\rangle),|BC \rangle, |CB \rangle \}.$ Under this basis, the $\sigma_i \, 's$ have the following image:

$\sigma_1 =
\begin{pmatrix}
-1 &  0  &  0 \\
0 &  1 &  0 \\
0 &  0  &  -1 \\
\end{pmatrix}
\qquad
\sigma_2 = -\frac{1}{2}
\begin{pmatrix}
0  & \sqrt{2} & -\sqrt{2}\\
\sqrt{2} & 1  &  1      \\
-\sqrt{2} & 1  &  1      \\
\end{pmatrix}
\qquad
\sigma_3 =
\begin{pmatrix}
-1 &  0  &  0 \\
0 &  -1 &  0 \\
0 &  0  &  1 \\
\end{pmatrix}
$

And they generate a group which is also isomorphic to $\Sym_4$.

Therefore, if we braid four anyons $C$, then all the images of the representations are very small.

\subsection{Representations on $V_z^{DDDD}$ }

There are six choices for $z$, namely $A$, $B$, $C$, $F$, $G$, $H$. By Remark \ref{identical particle}, we only need to consider cases where $z$ = $A$, $B$, $F$ and $G$.

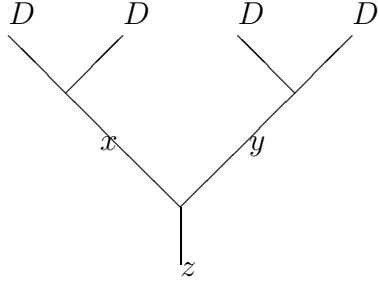
\setlength{\unitlength}{0.030in}
\begin{picture}(160,60)(0,-10)

\put(50,0){\line(0,1){10}}
\put(50,10){\line(1,1){30}}
\put(50,10){\line(-1,1){30}}
\put(70,30){\line(-1,1){10}}
\put(30,30){\line(1,1){10}}

\put(20,42){$D$}
\put(40,42){$D$}
\put(60,42){$D$}
\put(80,42){$D$}
\put(36,20){$x$}
\put(62,20){$y$}
\put(50,-2){$z$}
\end{picture}

\subsubsection{z = A}

The space $V_A^{DDDD}$ is five dimensional with basis $\{|AA\rangle, |GG\rangle,\\ |FF\rangle, |CC\rangle, |HH\rangle \}.$ Under this basis, the matrices of the $\sigma_i \, 's$ are as follow:

$\sigma_1 = \sigma_3 =
\begin{pmatrix}
1 &  0  &  0 & 0 &  0\\
0 &  \omega^2 &  0  & 0 &  0\\
0 &  0  & 1  & 0 &  0\\
0 &  0  & 0  & 1 &  0\\
0 &  0  & 0  & 0 &  \omega\\
\end{pmatrix}
$

$
\sigma_2 = \frac{1}{6}
\begin{pmatrix}
2                     &  \sqrt{2} + \sqrt{6}i   &      2\sqrt{2}     &  2\sqrt{2}     &  -\sqrt{2} - \sqrt{6}i\\
-\sqrt{2} + \sqrt{6}i &  1+ \sqrt{3}i           &   1- \sqrt{3}i     &  1- \sqrt{3}i  &   4                   \\
2\sqrt{2}             &  1- \sqrt{3}i           &    4               &   -2           &  1+ \sqrt{3}i\\
2\sqrt{2}             &  1- \sqrt{3}i           &    -2              &   4            &  1+ \sqrt{3}i\\
-\sqrt{2} - \sqrt{6}i &  4                      & 1+ \sqrt{3}i       & 1+ \sqrt{3}i   &  1- \sqrt{3}i\\
\end{pmatrix}
\qquad
$

where $\omega = -\frac{1}{2} + \frac{\sqrt{3}i}{2} $ is a third root of unity.

This representation splits into the direct sum of two trivial representations and a three-dimensional sector $S$, which is spanned by the basis $\{|GG\rangle, |HH\rangle, \frac{1}{2}(-\sqrt{2}|AA\rangle + |CC\rangle + |FF\rangle )\}.$ The representation on $S$ is generated by the following matrices:

$\sigma_1 = \sigma_3 =
\begin{pmatrix}
\omega^2   &  0  &  0 \\
0 &  \omega &  0 \\
0 &  0  &  1 \\
\end{pmatrix}
\qquad
\sigma_2 = \frac{1}{6}
\begin{pmatrix}
1+ \sqrt{3}i   &  4          &  2- 2\sqrt{3}i \\
4              &1- \sqrt{3}i &  2+ 2\sqrt{3}i \\
2- 2\sqrt{3}i  &2+ 2\sqrt{3}i  &  -2 \\
\end{pmatrix}
$

They generate a group of order 12 which is isomorphic to $\A_4$, the alternating group.

\subsubsection{$z = B$}

$V_B^{DDDD}$ is 4-dimensional with basis $\{ |GG\rangle, |FF\rangle, |CC\rangle, |HH\rangle \}.$ The matrices of the generators are:

$\sigma_1 = \sigma_3 =
\begin{pmatrix}
\omega^2 &  0  & 0 &  0\\
 0  & 1  & 0 &  0\\
 0  & 0  & 1 &  0\\
 0  & 0  & 0 &  \omega\\
\end{pmatrix}
$

$
\sigma_2 = \frac{1}{6}
\begin{pmatrix}
3-\sqrt{3}i  &  3+\sqrt{3}i  & 3+\sqrt{3}i   &  0\\
3+\sqrt{3}i  &   0           & -2\sqrt{3}    &  -3+\sqrt{3}i\\
3+\sqrt{3}i  &   -2\sqrt{3}  &  0            &  3-\sqrt{3}i\\
 0           &  -3+\sqrt{3}i & 3-\sqrt{3}i   &  3+\sqrt{3}i\\
\end{pmatrix}
$

This representation splits into the sum of two $2$-dimensional sectors $S_1$ and $S_2$, where $S_1$ is spanned by $\{|GG\rangle, \frac{1}{\sqrt{2}}(|CC\rangle+ |FF\rangle)\},$ and $S_2$ is spanned by $\{|HH\rangle, \frac{1}{\sqrt{2}}(|FF\rangle- |CC\rangle)\}.$

On the sector $S_1$, the generators have matrices:

$\sigma_1 = \sigma_3 =
\begin{pmatrix}
\omega &  0  \\
0 &  \omega^2  \\
\end{pmatrix}
\qquad
\sigma_2 =
\begin{pmatrix}
-\frac{1}{2} - \frac{\sqrt{3}i}{6}   & -\frac{\sqrt{6}i}{3}\\
-\frac{\sqrt{6}i}{3}                 & -\frac{1}{2} + \frac{\sqrt{3}i}{6} \\
\end{pmatrix}
$

They generate a group of size $24$ which is isomorphic to $\textrm{SL}(2, \mathbb{F}_3)$. Modulo the center, we get $\A_4$.

The representation on the other sector $S_2$ is exactly the same of that on $S_1$.

\subsubsection{$z = F$}

The space $V_F^{DDDD}$ is nine dimensional and has a basis $\{|FF\rangle, |AF\rangle, |FA\rangle, |GC\rangle, |CG\rangle, |GH\rangle, |HG\rangle, |CH\rangle, |HC\rangle \}.$ The generators of $\B_4$ have the following matrices:

$\sigma_1 =
\begin{pmatrix}
1 &    &   &  & &&&&&  \\
  &  1 &   &  & &&&&& \\
  &    & 1 &  & &&&&& \\
  &    &   &\omega^2  & &&&&& \\
  &    &   &          & 1  &&&&& \\
  &    &   &          &    & \omega^2&&&& \\
  &    &   &          &    &         & \omega &&& \\
  &    &   &          &    &         &        & 1 && \\
  &    &   &          &    &         &        &   & \omega & \\
\end{pmatrix}
$

$
\sigma_3 =
\begin{pmatrix}
1 &    &   &  & &&&&&  \\
  &  1 &   &  & &&&&& \\
  &    &1 &  & &&&&& \\
  &    &   &1  & &&&&& \\
  &    &   &          & \omega^2  &&&&& \\
  &    &   &          &    & \omega&&&& \\
  &    &   &          &    &         & \omega^2 &&& \\
  &    &   &          &    &         &        & \omega && \\
  &    &   &          &    &         &        &   & 1 & \\
\end{pmatrix}
$

$\sigma_2 = \frac{1}{3}
\begin{pmatrix}
1             & 1             & 1          & \omega^2  & \omega^2  &    1        &   1        & \omega   & \omega  \\
1             & 1             & 1          & \omega  &    1          & \omega&\omega^2&  1             &\omega^2  \\
1             & 1             & 1          &  1            & \omega  & \omega^2&\omega& \omega^2   &  1           \\
\omega^2  & \omega  & 1          &  1            & \omega^2  & \omega&  1         &  1             &   1          \\
\omega^2  & 1             &\omega& \omega^2  &  1            &  1          &\omega&  1             &  1           \\
1             & \omega  &\omega^2& \omega  &  1            &  1          &  1         & \omega^2    & 1         \\
1             & \omega^2  &\omega&   1           & \omega  &  1          &  1         &  1              &\omega^2 \\
\omega  & 1             &\omega^2&   1           &  1            & \omega^2&  1         &  1              &\omega \\
\omega  & \omega^2  & 1          &   1           &  1            &  1          &\omega^2& \omega    & 1           \\
\end{pmatrix}
$

This representation splits into the sum of a $1$-dim trivial representation and an $8$-dim irrep. The $1$-dim irrep is spanned by the element $\frac{1}{\sqrt{3}}(|FF\rangle + |AF\rangle + |FA\rangle).$ The $8$-dim irrep has an image in $\textrm{U}(8)$ of size $216$ which is isomorphic to the famous Hessian group $\sum(216) $ in physics literature. The following is a presentation of $\sum(216)$ from analyzing the matrices of the generators $\sigma_i$:

$$ <a,b,c | aba = bab, bcb = cbc, ac = ca, a^3 = (ab)^6 = (bc)^6 = (abcaba)^2 = 1>.
$$

As an abstract group, it is isomorphic to $((\mathbb{Z}_3 \times \mathbb{Z}_3) \rtimes Q_8) \times \mathbb{Z}_3,$ where $Q_8$ is the quaternion group of order $8$.

We will see this group again later.

\subsubsection{$z = G$}

The space $V_G^{DDDD}$ is also nine dimensional with a basis $\{|GG\rangle, |AG\rangle,\\ |GA\rangle, |FC\rangle, |CF\rangle, |FH\rangle, |HF\rangle, |CH\rangle, |HC\rangle\}$. As always, we first look at the matrices of the $\sigma_i \, 's$:

$\sigma_1 =
\begin{pmatrix}
\omega^2 &    &   &  & &&&&&  \\
  &  1 &   &  & &&&&& \\
  &    & \omega^2 &  & &&&&& \\
  &    &   & 1  & &&&&& \\
  &    &   &          & 1  &&&&& \\
  &    &   &          &    & 1&&&& \\
  &    &   &          &    &         & \omega &&& \\
  &    &   &          &    &         &        & 1 && \\
  &    &   &          &    &         &        &   & \omega & \\
\end{pmatrix}
$

$
\sigma_3 =
\begin{pmatrix}
\omega^2 &    &   &  & &&&&&  \\
  &  \omega^2 &   &  & &&&&& \\
  &    &1 &  & &&&&& \\
  &    &   &1  & &&&&& \\
  &    &   &          & 1  &&&&& \\
  &    &   &          &    & \omega&&&& \\
  &    &   &          &    &         & 1 &&& \\
  &    &   &          &    &         &        & \omega && \\
  &    &   &          &    &         &        &   & 1 & \\
\end{pmatrix}
$

$\sigma_2 = \frac{1}{3}
\begin{pmatrix}
1             & \omega             & \omega         & \omega^2  & \omega^2  &    1        &   1        & 1   & 1  \\
\omega             & 1             & \omega          & 1  &    1          & 1&\omega^2&  1             &\omega^2  \\
\omega             & \omega             & 1          &  1            & 1  & \omega^2&1& \omega^2   &  1           \\
\omega^2  & 1  & 1          &  1            & \omega^2  & 1&  \omega         &  \omega             &   1          \\
\omega^2  & 1             &1&\omega^2  &  1            &  \omega          &1&  1             &  \omega           \\
1             & 1  &\omega^2& 1  &  \omega            &  1          &  1         & \omega^2    & \omega         \\
1             & \omega^2  &1&   \omega           & 1  &  1          &  1         &  \omega              &\omega^2 \\
1  & 1             &\omega^2&   \omega           &  1            & \omega^2&  \omega         &  1              &1 \\
1  & \omega^2  & 1          &   1           &  \omega            &  \omega          &\omega^2& 1    & 1           \\
\end{pmatrix}
$

The representation splits into the sum of a 6-dim irrep and a 3-dim irrep.

Denote this 3-dim irrep by $W$, which is spanned by the basis $\{\frac{1}{\sqrt{2}}(|FC\rangle - |CF\rangle), \frac{1}{\sqrt{2}}(|CH\rangle - |FH\rangle), \frac{1}{\sqrt{2}}(|HF\rangle - |HC\rangle)\}.$ If we use this subspace $W$ as computational space, we will have a qutrit. The three basis elements above correspond to $|0\rangle, \,  |1\rangle, \,  |2\rangle. $

Under this basis, the matrices of the $\sigma_i \, 's$(unnormalized) on $W$ are as follows:

$\sigma_1 =
\begin{pmatrix}
1 &  0  &  0 \\
0 &  1 &  0 \\
0 &  0  &  \omega \\
\end{pmatrix}
\qquad
\sigma_3 =
\begin{pmatrix}
1 &  0  &  0 \\
0 &  \omega &  0 \\
0 &  0  &  1 \\
\end{pmatrix}
$

$\sigma_2 =
\begin{pmatrix}
\frac{1}{2} + \frac{\sqrt{3}i}{6} &  -\frac{1}{2} + \frac{\sqrt{3}i}{6}  &  -\frac{1}{2} + \frac{\sqrt{3}i}{6} \\
-\frac{1}{2} + \frac{\sqrt{3}i}{6} & \frac{1}{2} + \frac{\sqrt{3}i}{6}   &  -\frac{1}{2} + \frac{\sqrt{3}i}{6} \\
-\frac{1}{2} + \frac{\sqrt{3}i}{6} &  -\frac{1}{2} + \frac{\sqrt{3}i}{6}  &  \frac{1}{2} + \frac{\sqrt{3}i}{6} \\
\end{pmatrix}
$

 The group generated by them has order 648 with a center of size 3. The elements in the center are scalar matrices. And the group modulo center has order 216 which is isomorphic to the Hessian group $\sum(216).$

\subsection{Representations on $V_z^{GGGG}$}

The possible choices of $z$ are $A$, $B$ and $G$.

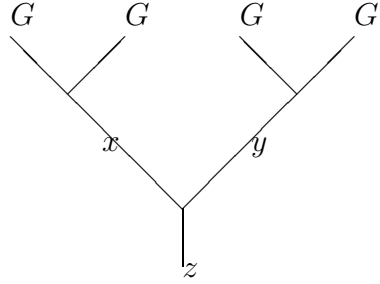
\setlength{\unitlength}{0.030in}
\begin{picture}(160,60)(0,-10)

\put(50,0){\line(0,1){10}}
\put(50,10){\line(1,1){30}}
\put(50,10){\line(-1,1){30}}
\put(70,30){\line(-1,1){10}}
\put(30,30){\line(1,1){10}}

\put(20,42){$G$}
\put(40,42){$G$}
\put(60,42){$G$}
\put(80,42){$G$}
\put(36,20){$x$}
\put(62,20){$y$}
\put(50,-2){$z$}
\end{picture}

\subsubsection{$z = A$}

The space $V_A^{GGGG}$ is 3-dimensional with a basis $\{|AA\rangle, \\ |BB\rangle, |GG\rangle\}.$ Under this basis, the matrices for the generators $\sigma_i$ $'s$ are given by

$\sigma_1 = \sigma_3 = \tau
\begin{pmatrix}
\omega^2 &  0  &  0 \\
0 &  -\omega^2 &  0 \\
0 &  0  &  \omega \\
\end{pmatrix}
$

$\sigma_2 = \tau
\begin{pmatrix}
\frac{1}{2}\omega & -\frac{1}{2}\omega  & \frac{1}{\sqrt{2}}\omega^2 \\
-\frac{1}{2}\omega & \frac{1}{2}\omega & \frac{1}{\sqrt{2}}\omega^2\\
\frac{1}{\sqrt{2}}\omega^2 & \frac{1}{\sqrt{2}}\omega^2 & 0 \\
\end{pmatrix}
$
where $\tau = e^{-\frac{\pi i}{9}}.$

This representation is irreducible and the group generated by them has a structure of $ (\mathbb{Z}_9 \times \mathbb{Z}_3) \rtimes \Sym_3$ with order 162, which is isomorphic to the group $D(9,1,1; 2,1,1)$.

We recall the definition of $D(n,a,b; d,r,s)$ below. For more information about this type of subgroups of $\textrm{SU}(3)$, see \cite{LP}.

Let

$E$ =
$
\begin{pmatrix}
0 & 1 & 0 \\
0 & 0 & 1 \\
1 & 0 & 0
\end{pmatrix}
\qquad
F = F(n,a,b) =
\begin{pmatrix}
e^{\frac{2\pi i a}{n}} & 0 & 0 \\
0 & e^{\frac{2\pi i b}{n}} & 0 \\
0 & 0 & e^{\frac{2\pi i (-a-b)}{n}}
\end{pmatrix}
$

$
G = G(d,r,s) =
\begin{pmatrix}
e^{\frac{2\pi i r}{d}} & 0 & 0 \\
0 & 0 & e^{\frac{2\pi i s}{d}} \\
0 & -e^{\frac{2\pi i (-r-s)}{d}}
\end{pmatrix}
$

Then $D(n,a,b;d,r,s)$ := $<E, F(n,a,b), G(d,r,s)>$.

Actually one can show that the group generated by the $\sigma_i \, 's$ is isomorphic to $D(9,1,1; 2,1,1)$ via a conjugation by some unitary matrix.

\subsubsection{$z = B$}

$V_B^{GGGG}$) is also 3-dimensional with a basis $\{|GG\rangle, |AB\rangle, |BA\rangle\}.$  The matrices of the $\sigma_i \, 's$ are given by

$\sigma_1  = \tau
\begin{pmatrix}
\omega &  0  &  0 \\
0 &  \omega^2 &  0 \\
0 &  0  &  -\omega^2 \\
\end{pmatrix}
\qquad
\sigma_3  = \tau
\begin{pmatrix}
\omega &  0  &  0 \\
0 &  -\omega^2 &  0 \\
0 &  0  &  \omega^2 \\
\end{pmatrix}
$

$\sigma_2 = \tau
\begin{pmatrix}
0 & -\frac{1}{\sqrt{2}}\omega^2 & -\frac{1}{\sqrt{2}}\omega^2\\
-\frac{1}{\sqrt{2}}\omega^2 &  \frac{1}{2}\omega & -\frac{1}{2}\omega\\
-\frac{1}{\sqrt{2}}\omega^2 &  -\frac{1}{2}\omega & \frac{1}{2}\omega \\
\end{pmatrix}
$
where $\tau = e^{-\frac{\pi i}{9}}.$

Again this representation is irreducible and they generate a group with structure $(\mathbb{Z}_{18} \times \mathbb{Z}_6) \rtimes \Sym_3,$ which is isomorphic to the group $D(18,1,1; 2,1,1)$ \cite{LP}. So it has order $648$.

Let
$
p= \begin{pmatrix}
0    &   0  &  1  \\
\frac{1}{\sqrt{2}} &  -\frac{1}{\sqrt{2}}   &  0 \\
\frac{1}{\sqrt{2}} &  \frac{1}{\sqrt{2}}    &  0 \\
\end{pmatrix}
$

Direct calculations show that conjugation by the matrix $p$ gives the isomorphism from our group generated by the $\sigma_i \, 's$ to $D(18,1,1;2,1,1)$.

\subsubsection{$z = G$}

$V_G^{GGGG}$ is now 5-dimensional with a basis $\{|GG\rangle, |AG\rangle, \\ |GA\rangle, |BG\rangle, |GB\rangle\}.$ The matrices of the generators $\sigma_i \, 's$ are:

$\sigma_1  =
\begin{pmatrix}
\omega &  0  &  0 & 0 &  0\\
0 &  \omega^2 &  0  & 0 &  0\\
0 &  0  & \omega  & 0 &  0\\
0 &  0  & 0  & -\omega^2 &  0\\
0 &  0  & 0  & 0 &  \omega\\
\end{pmatrix}
$

$
\sigma_2 =
\begin{pmatrix}
\omega &  0  &  0 & 0 &  0\\
0 &  \frac{1}{2}\omega &  \frac{1}{2}\omega^2  & \frac{1}{2}\omega &  -\frac{1}{2}\omega^2\\
0 &  \frac{1}{2}\omega^2  & \frac{1}{2}\omega  & -\frac{1}{2}\omega^2 &  \frac{1}{2}\omega\\
0 &  \frac{1}{2}\omega  & -\frac{1}{2}\omega^2  & \frac{1}{2}\omega &  \frac{1}{2}\omega^2\\
0 &  -\frac{1}{2}\omega^2  & \frac{1}{2}\omega  & \frac{1}{2}\omega^2 &  \frac{1}{2}\omega\\
\end{pmatrix}
\qquad
$

$\sigma_3  =
\begin{pmatrix}
\omega &  0  &  0 & 0 &  0\\
0 &  \omega &  0  & 0 &  0\\
0 &  0  & \omega^2  & 0 &  0\\
0 &  0  & 0  & \omega &  0\\
0 &  0  & 0  & 0 &  -\omega^2\\
\end{pmatrix}
$

It's obvious that $|GG\rangle$ is a common eigenvector of the $\sigma_i \, 's$. So it spans a 1-dim irrep of $\B_4$. The orthogonal complement spanned by the other $4$ basis elements is a 4-dim irrep.

The group generated by the $\sigma_i \, 's$ has order 648. And GAP shows that it has a structure of $(((\mathbb{Z}_3 \times ((\mathbb{Z}_3 \times \mathbb{Z}_3) \rtimes \mathbb{Z}_2)) \rtimes \mathbb{Z}_2 ) \rtimes \mathbb{Z}_3) \rtimes \mathbb{Z}_2$.

\bibliographystyle{plain}
\bibliography{mybib_UQGPFA}

\end{document}